\documentclass{article}

\usepackage{amssymb, amsmath, amsthm, verbatim, graphicx}

\begin{document}

\title{The Coalescence of Intrahost HIV Lineages Under Symmetric CTL Attack}
\author{Sivan Leviyang\\ 
Georgetown University\\  
Department of Mathematics}
\maketitle
  
\abstract{Cytotoxic T lymphocytes (CTLs) are immune system cells that are
thought to play an important role in controlling HIV infection. We develop a
stochastic ODE model of HIV-CTL interaction that extends current deterministic
ODE models.  Based on this stochastic
model, we consider the effect of CTL attack on intrahost HIV lineages assuming
CTLs attack several epitopes with equal strength.  
In this setting, we introduce a limiting version of our stochastic ODE under which we show that
the coalescence of HIV lineages can be described by a simple paintbox construction.   
Through numerical experiments, we show that our results under the limiting
stochastic ODE accurately reflect
HIV lineages under CTL attack when the HIV population size is on the low end
of its hypothesized range.  Current techniques of HIV lineage construction
depend on the Kingman coalescent.  Our results give an explicit connection
between CTL attack and HIV lineages.}

\newcommand{\nume}{\textbf{e}}
\newcommand{\alle}{\mathcal{E}}
\newcommand{\tb}{\textbf}
\renewcommand{\P}{\mathcal{P}}
\newcommand{\Tcell}{\text{CD}4^+}
\newcommand{\important}{\textbf}
\renewcommand{\H}{\mathbb{H}}
\newcommand{\E}{\mathbb{E}}
\newcommand{\spl}{\text{SPL}}
\newcommand{\spls}{\text{SPL }}
\newcommand{\cT}{T^\text{class}}
\newcommand{\hT}{T^\text{class,$h$}}
\newcommand{\ce}{e^\text{class}}
\newcommand{\full}{\text{full}}
\newcommand{\linear}{\text{linear}}
\newcommand{\eigconstant}{\rho}
\newcommand{\LD}{\text{LD}}
\newcommand{\classic}{\text{classic}}
\newcommand{\SP}{\text{SP}}
\newcommand{\Gin}{\frac{1}{r'}\log(\frac{w}{\lambda})}
\newcommand{\delk}{\Delta k}
\newcommand{\rr}{\frac{\delk}{r}}
\newcommand{\sample}{\text{sample}}
\newcommand{\Var}{\mathcal{V}}
\newcommand{\lin}{\ell}
\newcommand{\Nn}{\mathbb{N}_n}
\renewcommand{\L}{\mathcal{L}}
\newcommand{\color}{\text{color}}
\newcommand{\Gauss}{\mathcal{N}}
\newcommand{\M}{\mathcal{M}}
\newcommand{\splname}{single spawn limit}
\newcommand{\system}{(\ref{E:final_system})}
\newcommand{\systems}{(\ref{E:final_system}) }
\renewcommand{\S}{\mathcal{S}}
\newcommand{\error}{\text{error}}
\newcommand{\all}{\text{all}}
\newcommand{\est}{\text{estimate}}

\newtheorem{claim}{Theorem} 
\newtheorem{definition}{Definition}
\newtheorem{lemma}{Lemma}
\newtheorem{proposition}{Proposition}

\renewcommand{\theequation}{\arabic{section}.\arabic{equation}}
\renewcommand{\theclaim}{\arabic{section}.\arabic{claim}}
\renewcommand{\thelemma}{\arabic{section}.\arabic{lemma}} 

\section{Introduction}

Cytotoxic T lymphocytes (CTLs) are immune system cells that kill pathogen
infected host cells.  In the context of HIV infection,
considerable experimental evidence suggests that CTLs play a central role in
controlling infection and shaping HIV diversity, e.g.
\cite{Borrow_1994_J_Virology, Carrington_2003_AIDS, Goonetilleke_2009_JEM,
Koup_1994_J_Virology, Schmitz_1999_Science}.

Roughly, when HIV enters a host cell, typically a $\Tcell$ cell, certain
mechanisms within the cell cut up HIV proteins into small pieces (usually
$8-11$ amino acids long) and present these peptides on the surface of the cell
in the form a peptide-MHC complex (pMHC) \cite{deFranco_book}.  CTLs can bind to
pMHC complexes and then destroy the presenting cell, but critically each CTL
possesses receptors that can bind to a limited pattern of peptides.  An HIV
peptide that is attacked by CTLs is referred to as an \textit{epitope}.

When CTLs attack a given epitope, HIV infected cells possessing that
epitope are killed off.  However, due to its high mutation rate, many variants
of HIV exist during any moment of infection.   As a result,
infected cells possessing HIV variants that do not produce the attacked epitope
may exist prior to CTL attack or arise during the attack.    Such variants, which
are at a selective advantage due to the CTL attack,  will proliferate and come
to dominate the HIV population.  This hypothetical picture has been confirmed in
many experimental HIV studies, e.g. \cite{Kelleher_2001_JEM}.   Yet despite the
putative role of CTLs in  controlling HIV infection and the corresponding
importance of HIV genetic  diversity in evading CTL attack, the impact of CTL
attack on intrahost HIV genetic diversity is not well understood.

Most current theoretical tools used in HIV research do not link CTL models to
HIV genetic diversity. On one hand,
HIV-CTL interaction has been modeled since the beginning of the HIV epidemic  (see
\cite{Nowak_and_May_Book, Perelson_Nature_Reviews_2002} for a review). 
Various models are possible,  but the standard model consists of a deterministic ODE
composed of variables for the population size of HIV virions,  infected and
uninfected $\Tcell$ cells, and CTLs targeting infected $\Tcell$ cells.   While
the standard model and its many variations give a dynamic picture of HIV and
CTL population sizes, they do not connect CTL attack to HIV population genetics.   

On the other hand, tools from population genetics that do not explicitly
model CTL attack have been applied to HIV.    Rodrigo and
coworkers used variants of the Kingman coalescent to explore the HIV life
cycle and construct inference algorithms based on HIV genetic samples 
\cite{Rodrigo_Book_Chapter, Rodrigo_1999_PNAS, Drummond_2000_Mol_Bio_Evol}.  The
popular programs BEAST and LAMARC, which are used to make statistical
inferences based on HIV genetic data, assume a Kingman coalescent \cite{BEAST,
LAMARC}.

In this work, we present results that connect an ODE model of
HIV population dynamics under CTL attack to HIV population genetics.  More
specifically, we consider a stochastic ODE that models HIV escape from CTL attack at
multiple epitopes sometime during the chronic phase of infection.   Our
stochastic ODE describes the dynamics of the HIV population in terms of discrete birth, death, and mutation events, allowing us
to specify lineages once the dynamics are given. We show that under a certain
small population limit our stochastic ODE connects to the deterministic ODE
models described above.

To connect to HIV population genetics, we
consider a collection of HIV infected cells sampled after HIV has
escaped CTL attack.  Given a realization of the stochastic ODE dynamics,  the
lineages of these infected cells can be traced back to the time at which
CTL attack initiates, thereby forming a genealogy.   For simplicity, we assume
CTL attack of equal strength at each considered epitope, a situation we refer to
as symmetric attack.   In this setting, our main result characterizes the state
of the genealogy at the time when CTL attack initiates.   Further, we show that
HIV escape mutations produce significant stochasticity in the HIV population
dynamics.

We analyze our stochastic ODE using methods similar to those used by Iwasa, Michor,
Komarova, and Nowak \cite{Iwasa_JTB_2005} and Durrett,
Schmidt, and Schweinsberg \cite{Durrett_Ann_App_Prob_2009} in their study of
cancer pathways. Hermisson and Pennings \cite{Hermisson_2005_Genetics,
Pennings_MBE_2006} also used similar techniques in an abstract setting
applicable to HIV.   In all these works and our own,  the dynamics of mutations
present at low levels in the overall population  are well approximated by
branching processes. Rouzine and Coffin considered an HIV model that bears
some similarity to our HIV-CTL model \cite{Rouzine_2010_TPB},  but their
analysis and goals differ from ours.

Our lineage construction is similar in spirit to that of several authors, but
there are significant differences between our underlying model and that of
previous authors.   In \cite{Kaplan_Genetics_1988, Durrett_TPB_2004} the authors
considered lineages from a population that has undergone a strong selective
sweep,  while in \cite{Barton_Ann_App_Prob_2004}  the authors considered
lineages from a population under selection-mutation equilibrium.   Both these
works considered a fixed size, Moran model with weak mutation rates.   In our
case, the stochastic ODE considered does not assume a fixed population size and 
we consider a strong mutation rate reflective of HIV biology.  

In section \ref{S:model} we describe our model.  In section \ref{S:results} we
describe our theoretical results along with associated numerical results.  In
section \ref{S:discussion} we discuss some implications of our results.  
Sections \ref{S:linear_escape_graph} and \ref{S:full_escape_graph} provide
proofs of the results presented in Section \ref{S:results}.   In these two
sections, we have endeavored to focus on the intuition behind the proofs.  Our
hope is that the mathematics presented in these sections contributes to
intuition and biological motivation.  We place arguments that are mathematically
technical, and unnecessary for intuition, in the appendix.

\section{A Model of HIV Dynamics Under CTL Attack} \label{S:model}
\setcounter{equation}{0}
\setcounter{claim}{0}

To specify our model, in section \ref{S:escape_graph} we
introduce terminology that will help characterize the CTL attack. In section \ref{S:ODE}, we introduce
our stochastic ODE model and connect it to a deterministic ODE similar to those
mentioned in the introduction. In section \ref{S:symmetric}, we specify
a specific parameter choice for our stochastic ODE that models symmetric CTL
attack. Finally, in section \ref{S:genealogies}, we discuss genealogies within the
context of our HIV-CTL model.

\subsection{Escape Graph}  \label{S:escape_graph}

We model an HIV population exposed to attack at $\nume$ epitopes.  To do this, 
we categorize the HIV infected cells by the presence, represented by a $\tb{0}$,
or absence, represented by a $\tb{1}$, of a given epitope.   Since there are
$\nume$ epitopes, the different HIV infected cell variants can be associated
with a binary number of length $\nume$. For example if $\nume=2$, then the
HIV infected cell variant, hereafter we simply say variant, $\tb{10}$
represents an infected cell containing only  the second epitope. Intuitively, we
think of $\tb{1}$'s as representing mutations that alter a gene on the HIV
genome that is responsible for producing the attacked epitope.

We let $\alle$ be the set of possible variants.   In this work, we focus on two
possible choices for $\alle$. In the first, which we label as $\alle_\full$, we
consider every possible combination of epitopes.  For example, if $\nume=3$ we define
\begin{equation}
\alle_\full = \{\tb{000}, \tb{001}, \tb{010}, \tb{011}, \tb{100}, \tb{101},
\tb{110}, \tb{111}\}.
\end{equation}
In the second choice, which we label as $\alle_\linear$, we consider only
variants that, from left to right, contain a sequence of all $\tb{1}$'s followed
by a sequence of all $\tb{0}$'s.  In the case $\nume=3$ we define
\begin{equation}
\alle_\linear = \{\tb{000}, \tb{100}, \tb{110}, \tb{111}\}
\end{equation}

Given $\alle$ we define a graph $G$ which we call the \textit{escape graph}
of $\alle$.  $G$ is formed from vertices labeled by elements of $\alle$ and
arrows that connect a vertex with label $v'$ to one with label $v$ if a single
epitope mutation can change $v'$ to $v$.   When $\alle = \alle_\full$ and
$\alle = \alle_\linear$  we refer to the associated $G$ as the \textit{full and
linear escape graph}, respectively.  Figures \ref{F:full_escape_graph} and
\ref{F:linear_escape_graph} show the full and linear escape graph, respectively,
in the case $\nume=3$.

\begin{figure} [h]
\begin{center} 
\includegraphics[width=1\textwidth]{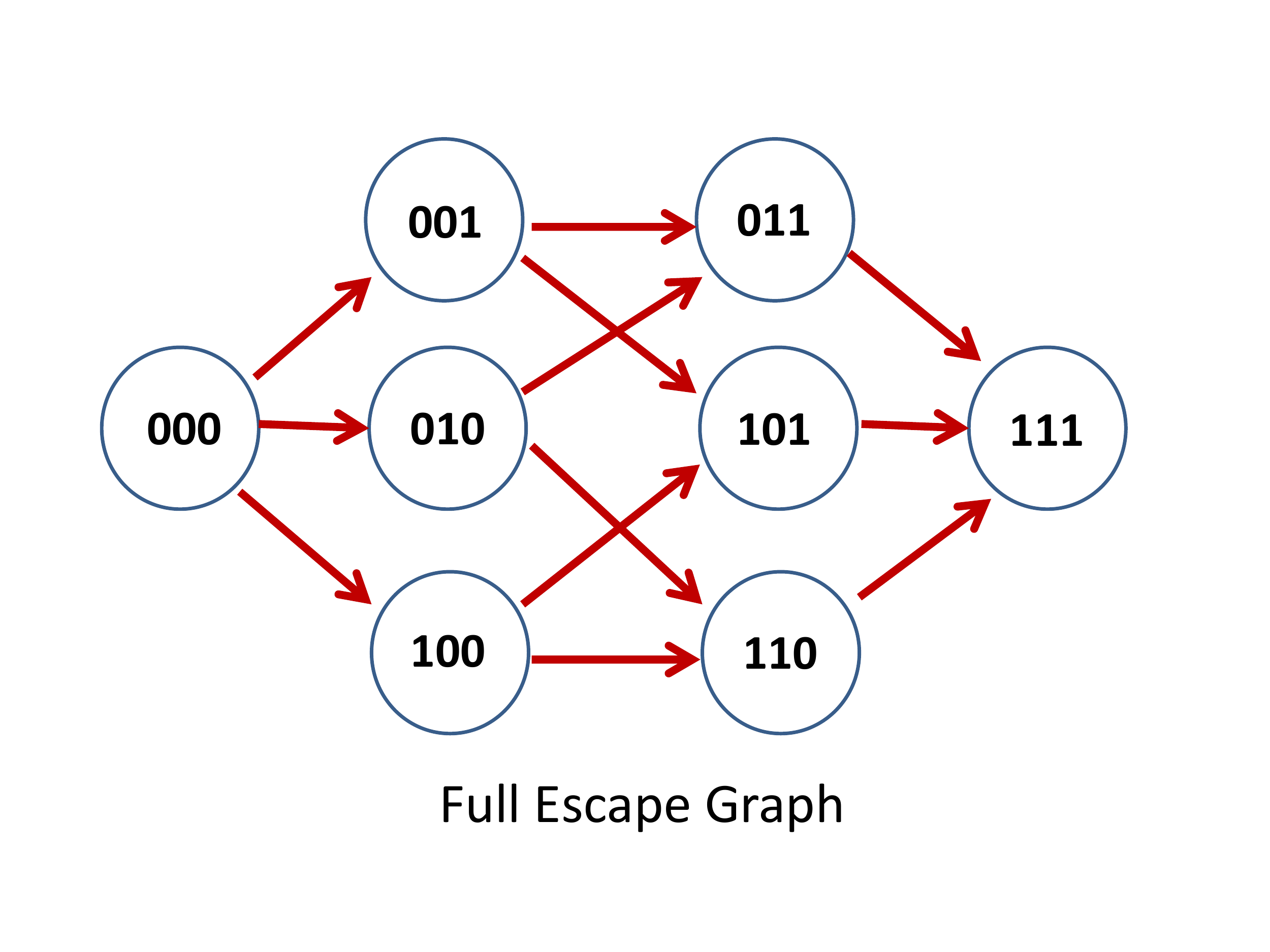}
\caption{Full Escape Graph for $\nume = 3$.} 
\label{F:full_escape_graph}
\end{center}
\end{figure}
 
\begin{figure} [h]
\begin{center} 
\includegraphics[width=1\textwidth]{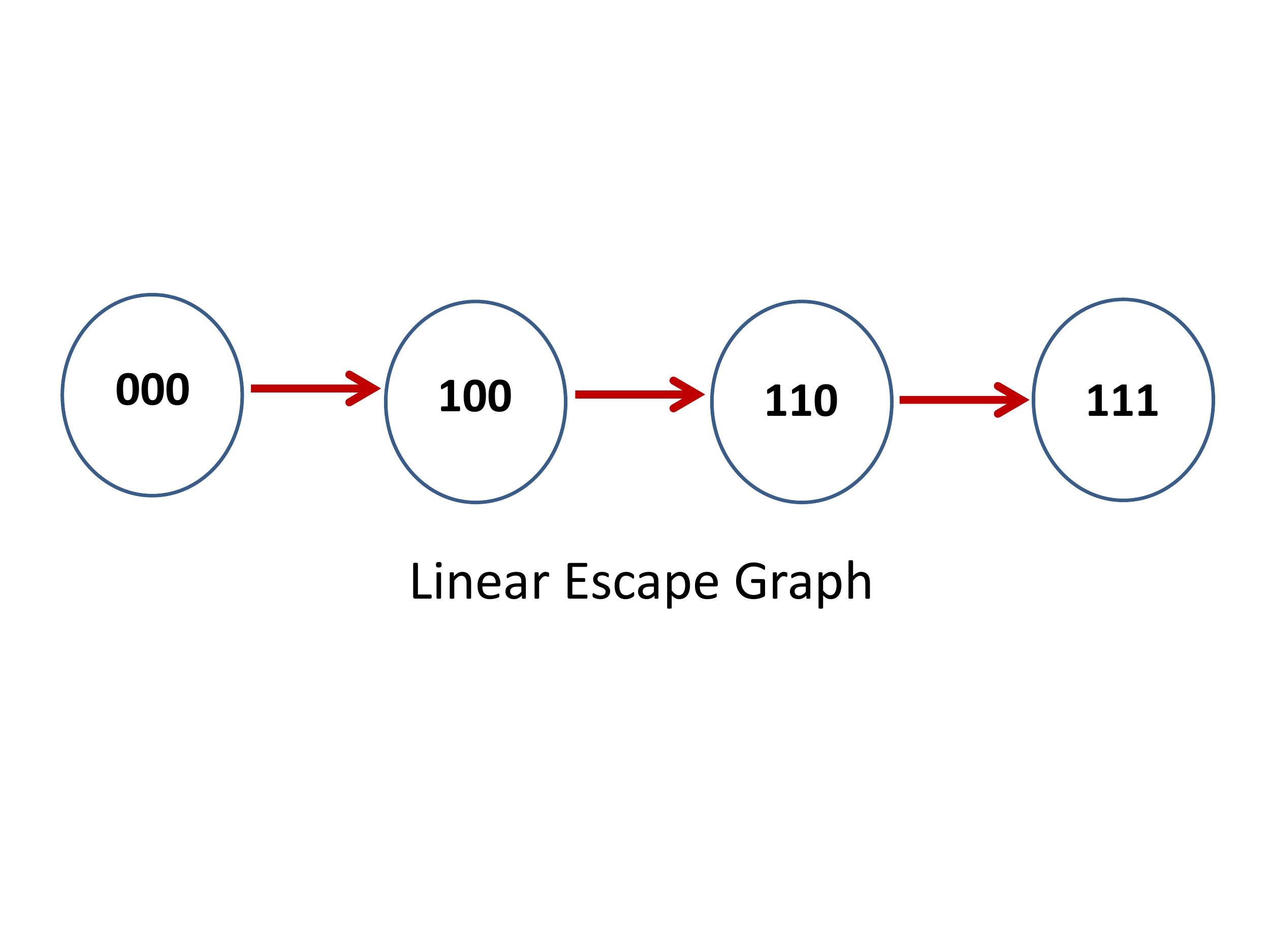}
\caption{Linear Escape Graph for $\nume = 3$.} 
\label{F:linear_escape_graph}
\end{center}
\end{figure}
 
For any $v \in \alle$, $\P(v)$ is the set of elements in $\alle$ that can
be changed into $v$ by transforming a single $\tb{0}$ into a $\tb{1}$.  
Intuitively, we think of $\P(v)$ as the variants that can be transformed
into $v$ by a single mutation and the $\P$ stands for parents.  To be clear, if
$v=\tb{110}$ then we have $\P(v) = \{\tb{010, 100}\}$ and $\P(v) = \{\tb{100}\}$
for the full and linear escape graphs respectively.

We say that the HIV population has escaped CTL attack when all infected cells
are of type $\tb{111\dots1}$.  In other words, mutations that remove each of the
attacked epitopes have fixed in the HIV population.

\subsection{ODE}  \label{S:ODE}

Let $h$ represent the
 number of uninfected $\Tcell$ cells that are targets for HIV
infection.    For each $v \in \alle$ let $e_v$ be the number of $\Tcell$ cells 
infected by  $v$ variants.  We assume birth and death rates for $h$  and the
$e_v$ as  specified in Table \ref{T:rates}. $\lambda$ and $g$ represent  the
birth and  death rates respectively of a $\Tcell$ cell in the absence of  HIV
infection.  $b_v h$ and $k_v$ are the birth and death rate of a $v$ variant,
infected $\Tcell$  cell.  An infected $\Tcell$ birth event corresponds to an
uninfected $\Tcell$  death event, so uninfected $\Tcell$ cells have an
additional death term beyond $g$.

We let $\mu$ be the rate per infection event at which mutations occur that
remove any one of the $\nume$ epitopes.   Correspondingly, new $v$
variants arise from mutations in $v' \in \P(v)$ with a
rate given in the 'mutation event' rate column in Table \ref{T:rates}.  Notice
that a 'mutation event' in Table \ref{T:rates} refers to the creation
of a $v$ variant from a mutation in some variant contained in $\P(v)$, not the mutation
of a $v$ variant itself.

\begin{table}
\begin{center} 
\begin{tabular}{|c||c|c|c|}  
\hline 
cell type ($\#$ of cells) & birth rate & death rate & mutation event \\
\hline
uninfected ($h$)  & $\lambda$ & $g + \sum_{v \in \alle} b_v e_v$ & -
\\
\hline
$v$ infected ($e_v$) & $b_v h$ & $k_v$ & $\mu
\sum_{v' \in \P(v)} b e_{v'} h$
\\
\hline
\end{tabular}
\end{center}
\caption{Birth, Death, and Mutation Rates}
\label{T:rates}
\end{table}
  
	Define $P(f(t))$ to be a Poisson process with jump rate $f(t)$ at time $t$. 
Then, given the rates in Table \ref{T:rates} we have the following stochastic
ODE,
\begin{gather} \label{E:unscaled_system}
dh = dP(\lambda) - dP(gh) - \sum_{v' \in \alle} dP(b_{v'} e_{v'} h) \\ \notag
de_v = dP(b_v e_v h) - dP(k_v e_v) + \sum_{v' \in \P(v)} dP(\mu b_{v'} e_{v'}
h).
\end{gather}
where the second equation directly above applies for all $v \in \alle$.   Each
$P$ in (\ref{E:unscaled_system}) represents an independent Poisson process run
at the specified rate, to avoid cumbersome notation we do not use a distinct
notation for each of these processes.  In (\ref{E:unscaled_system}) and
throughout this work, we ignore back mutations, a mutation from a variant to
a less fit variant.  Ignoring such mutations does not affect
our results.

To make our system variables ($h$ and the
$e_v$) $O(1)$, we rescale $h$ and each $e_v$ by $\H$ and $\E$ respectively. 
Intuitively, $\H$ and $\E$ correspond to the order at which uninfected
but infectable $\Tcell$ cells and infected, activated $\Tcell$ cells capable of
producing virions exist during HIV infection, respectively. We set $\H = \lambda/g$ , the
steady state of uninfected cells in the absence of HIV infection.  This scaling is supported by empirical results suggesting that, at
least prior to AIDS onset, the number of uninfected $\Tcell$ cells during and
prior to HIV infection are on the same order \cite{Evolution_HIV_Book}.  
Without justification for a moment, we choose $\E = g/\bar{b}$ where $\bar{b}$ is on the order of the $b_v$.  If we rewrite $h$ and $e_v$ as $h/\H$ and $e_v/\E$
respectively in (\ref{E:unscaled_system}), we arrive at the rescaled system,
\begin{gather} \label{E:unscaled_system_2}
dh = \frac{g}{\lambda} \left(dP(\lambda) - dP(\lambda h) 
	- \sum_{v' \in \alle} dP(\lambda \frac{b_{v'}}{\bar{b}} e_{v'} h) \right) \\
	\notag 
de_v = \frac{\bar{b}}{g} \left( dP(\lambda \frac{b_v}{\bar{b}} e_v h) 
	- dP(\left(\frac{k_v g}{\bar{b}}\right) e_v) 
	+ \sum_{v' \in \P(v)} dP(\mu \lambda  \frac{b_{v'}}{\bar{b}} e_{v'} h) \right).
\end{gather}

We would like to recover a deterministic ODE from (\ref{E:unscaled_system_2}),
in this way showing that our present model is an extension of current
deterministic models.   In \cite{Kurtz_Book}, Kurtz showed that one can recover
deterministic population ODEs by taking large population limits of stochastic
population ODEs.   In our context, we can consider $\H \to \infty$ and $\E \to
\infty$.   Such limits are reasonable for HIV due to its enormous population
size, but it is not immediately clear what the relationship should be between
$\H$ and $\E$ as both go to infinity.

To address this issue in a simple context, consider (\ref{E:unscaled_system_2})
without CTL attack.  In this setting we need not distinguish between different
variants, reducing our system to the variables $h$ and $e$, and we may also
ignore mutation.  Taking $\H$ and $\E$ large,   we can largely ignore the
stochasticity of (\ref{E:unscaled_system}) and arrive at the following
deterministic ODE,
\begin{gather} \label{E:unscaled_system_det}
\frac{\text{d}h}{\text{dt}} = g (1 - h - e h) \\ \notag 
\frac{\text{d}e}{\text{dt}} = \frac{b}{g} e (\lambda h - \frac{k g}{b}),
\end{gather}
which has the equilibrium,
\begin{equation} \label{E:h_equil} 
h = (k g)/(b \lambda)
\end{equation}  
Consider the variables in the expression for $h$ directly above. $k$, the death
rate of infected $\Tcell$ cells has been measured  at approximately $2$ days
\cite{Perelson_Science_1996}. If we take $2$ days as our time scale, we then 
expect $k \approx 1$.  Uninfected $\Tcell$ cells last on the order  $2$ weeks,
giving $g = .1$ as a reasonable choice.   Estimates for $b$ and $\lambda$ have
significant variation in the literature.   However, we can understand the role
of $\lambda$ and $b$ in (\ref{E:unscaled_system_2}) by noting the following relation,
\begin{equation}
\lambda b = g^2 \left(\frac{\H}{\E}\right),
\end{equation}
which follows from our formulas for $\H$ and $\E$.  The above relation and
(\ref{E:h_equil})  demonstrate that (\ref{E:unscaled_system_2}) only converges
to a deterministic  system in the large population limit of $\H, \E \to \infty$
if the ratio $\H/\E$ converges to a fixed constant.  

To force (\ref{E:unscaled_system_2}) to have a deterministic limit, we introduce
a parameter $\gamma = \lambda b/g$ or in terms of $\H, \E$, $\gamma = g(\H/\E)$
(the factor $g$ is not essential, but gives the system directly below a
simpler form).  From a biological point of view, $\gamma$ is an inverse measure
of the fraction of infectable cells that are actually infected.   Empirical
results for the ratio of infected to infectable cells are difficult as most
infected $\Tcell$ are in the lymph nodes and many such $\Tcell$ are infected but
inactive \cite{Levy_HIV_Book}.  However, estimates in the range of
$.01$ to $.1$ have been given by several authors and seem
reasonable \cite{Levy_HIV_Book}. With $g=.1$, the corresponding  range for
$\gamma$ is $1$ to $10$.  Using this scaling of $\gamma$, our definition of $\E$
is justified biologically.  

Rewriting (\ref{E:unscaled_system_2}) using $\gamma$ and setting $\tilde{b}_v =
b_v/\bar{b}$ gives
\begin{gather} \label{E:unscaled_system_3}
dh = \frac{1}{\E}\left(\frac{g}{\gamma}\right) \left(dP(\gamma \E) -
dP(\gamma \E h) - \sum_{v' \in \alle} dP(\tilde{b}_{v'} \gamma \E e_{v'} h)
\right) 
\\ \notag 
de_v = \frac{1}{\E} \left(
 dP(\gamma \tilde{b}_v \E e_v h) - dP(k_v \E e_v) 
	+ \sum_{v' \in \P(v)} dP(\mu \tilde{b}_{v'} \gamma \E e_{v'} h) \right).
\end{gather}
and we consider the limit of this system as $\E \to \infty$.   

\textit{While (\ref{E:unscaled_system_3}) approaches a deterministic limit as
$\E \to \infty$ when mutation is ignored, the system will continue to be
stochastic if $\mu$ is sufficiently large with respect to $\E$}.   Indeed, as we
show in section \ref{S:results}, the scaling of $\mu$ that produces stochasticity is
precisely a scaling in which HIV lives.  
Roughly, stochasticity of (\ref{E:unscaled_system_3}) exists even as $\E \to \infty$ because the dynamics
of variants that are of scaled population size $O(\frac{1}{\E})$, i.e. $e_v =
O(\frac{1}{\E})$,  will be stochastic even as $\E \to \infty$. 

However, if we ignore mutation then as $E \to \infty$,
(\ref{E:unscaled_system_3}) becomes
\begin{gather}  \label{E:almost_PNN}
\frac{\text{d}h}{\text{d}t} = g(1-h - \sum_{v' \in \alle}  \tilde{b}_{v'} 
e_{v'} h) 
\\ \notag 
\frac{\text{d}e_v}{\text{d}t} = \gamma(\tilde{b}_v h e_v - k_v e_v), 
\end{gather}
which has the form of a predator-prey system.   (\ref{E:almost_PNN}) is a
simplified version of the standard deterministic ODE used today for HIV
modeling, the full version includes the virion population.   But generally, the
reduction to (\ref{E:almost_PNN}) demonstrates how (\ref{E:unscaled_system_3}) 
is based on current HIV models.

\subsection{Symmetric CTL Attack and Initial Conditions}  \label{S:symmetric}

We consider (\ref{E:unscaled_system_3}) restricted to the case of symmetric
attack. To make the notion of symmetric attack precise, we partition the 
collection of variants, $\alle$, into subsets $\alle_i$ such that 
\begin{equation}
\alle_i = \{v \in \alle : v \text{ has i $1$'s in its binary expression}\}
\end{equation}
For example,  if $\nume = 4$, then $\alle_1 = \{\tb{1000, 0100, 0010, 0001}\}$
and $\alle_1 = \{\tb{1000}\}$ for the full and linear escape graph,
respectively.  $\alle_i$ is the collection of variants that are mutated at $i$
epitopes.  We refer to the $\alle_i$ generally as \textit{variant classes} and
$\alle_i$ specifically as the \textit{$i$th variant class}.

To model symmetric attack, we assume that a variant  $v \in
\alle_i$ will be exposed to CTL attack at $\nume-i$ epitopes, and we assume
that the death rate due to CTLs at each single epitope has rate
$\delk$.  We scale time so that infected cells die, in the absence of CTL
attack, at rate $1$.  As mentioned, the lifetime of an infected cell has been
shown to be approximately $2$ days which is in turn our unit of time.  All this
is made precise by taking the death
rate $k_v$ of $v \in \alle_i$ to be given by
\begin{equation}
k_v = 1 + (\nume-i)\delk.
\end{equation}
Finally, as an added simplification, we take $b_v$ to
be constant. In (\ref{E:unscaled_system_3}) this amounts to taking $\tilde{b}_v
= 1$.

Before presenting our final system, we note that mutations do not play a role
in the equation for $h$ and so stochastic effects will have little impact on $h$
dynamics.   For simplicity, and with error that goes to $0$ as $\E \to \infty$,
we replace the $h$  equation by its deterministic counterpart.  Putting all
these remarks together, we arrive at the following system.
\begin{gather} \label{E:final_system}
\frac{\text{d}h}{\text{d}t} = g \left(1 - h -  \sum_{v' \in \alle} e_{v'} h
\right) 
\\ \notag 
de_v = \frac{1}{\E} \left(dP(\gamma \E e_v h) - dP(k_v \E e_v) 
	+ \sum_{v' \in \P(v)} dP(\mu \gamma \E e_{v'} h) \right).
\end{gather}
From this point on, we take (\ref{E:final_system}) as describing the
dynamics of the HIV population.  

To set initial conditions, we assume that variant $v_0 = \tb{000\dots0}$ is
the dominant variant prior to CTL attack.  Indeed, CTLs proliferate in response
to epitopes existing in the population, so taking $v_0$ to be the dominant
variant is biologically reasonable.   Ignoring other variants for a moment, we
set $h, e_{v_0}$ at time $t=0$ according to the equilibrium of
(\ref{E:unscaled_system_det}):
\begin{gather}
h(0) = \frac{1}{\gamma}, \\ \notag
e_{v_0}(0) = \frac{1 - h(0)}{h(0)}.
\end{gather}
Prior to CTL attack, we assume that other variants are at a slight fitness
disadvantage to $v_0$ and arise due to  mutations on $v_0$
variants. Assuming, as we have just done, that there are $O(\E)$ $v_0$ variants,
there will be $O(\mu \E)$ variants from each of the $\alle_1$ classes.    From
this, we can conclude that the number of variants in the $\alle_2$ class will be
of order $O(\mu^2 \E)$.  As we mention below, $\mu^2 \E \approx 0$ and so we
assume that no $\alle_i$ variants exist at $t=0$ for $i>1$.  For simplicity we
assume $e_v(0) = \mu \E$ for all $v \in \alle_1$.

These initial conditions are not essential to our results, other choices are
possible.  Which initial conditions are appropriate will depend on the period
of HIV infection one has in mind.  We have made a specific choice
 for the sake of clarity.

\subsection{Genealogies}  \label{S:genealogies}

When all variants are of
type $\underbrace{\textbf{11\dots1}}_{\nume}$, the HIV population has escaped
CTL attack.  We let $T_\sample$ be a time after such an escape has been
completed and consider $n$ infected cells sampled at $T_\sample$. 
Since (\ref{E:final_system})  defines discrete birth and death events, we can
construct lineages corresponding to the ancestral lines of these $n$ sampled cells.

We label the lineages $\lin_1, \lin_2, \dots, \lin_n$ and we let $\Pi(t)$
represent the partition structure of the lineages at time $t$.  To explain this,
consider Figure \ref{F:partition} which represents a possible lineage structure
for the case $n=8$.  The values of $\Pi(t)$ at $t=T_\sample, T_B, T_C$ are given
by,
\begin{flushleft}
$\Pi(T_\sample) = \{\{\lin_1\}, \{\lin_2\}, \dots, \{\lin_8\}\}$ \\
$\Pi(T_B) = \{\{\lin_1, \lin_2\}, \{\lin_3, \lin_4\}, \{\lin_5, \lin_6\},
\{\lin_7, \lin_8\}\}$
\\
$\Pi(T_C) = \{\{\lin_1, \lin_2, \lin_3, \lin_4, \lin_5, \lin_6\},
\{\lin_7, \lin_8\}\}$
\end{flushleft}
At time $T_\sample$ all lineages are separate, $\Pi(T_\sample)$ consequently
partitions each lineage into its own set.  By time $T_B$, the pairs of lineages
$1$ and $2$, $3$ and $4$, $5$ and $6$, and $7$ and $8$ have coalesced. 
$\Pi(T_B)$ partitions these pairs to reflect this structure.  Finally by time
$T_C$, lineages $1$ through $6$ have coalesced as has the pair $7$ and $8$.  
$\Pi(T_C)$ partitions the lineages accordingly.

$\Pi(t)$ is a random partition function that encodes the genealogy formed by
the $n$ lineages. Its stochasticity follows from the stochasticity of
(\ref{E:final_system})  as well as the stochasticity of lineages given a single
realization of (\ref{E:final_system})

\begin{figure} [h]
\begin{center} 
\includegraphics[width=1\textwidth]{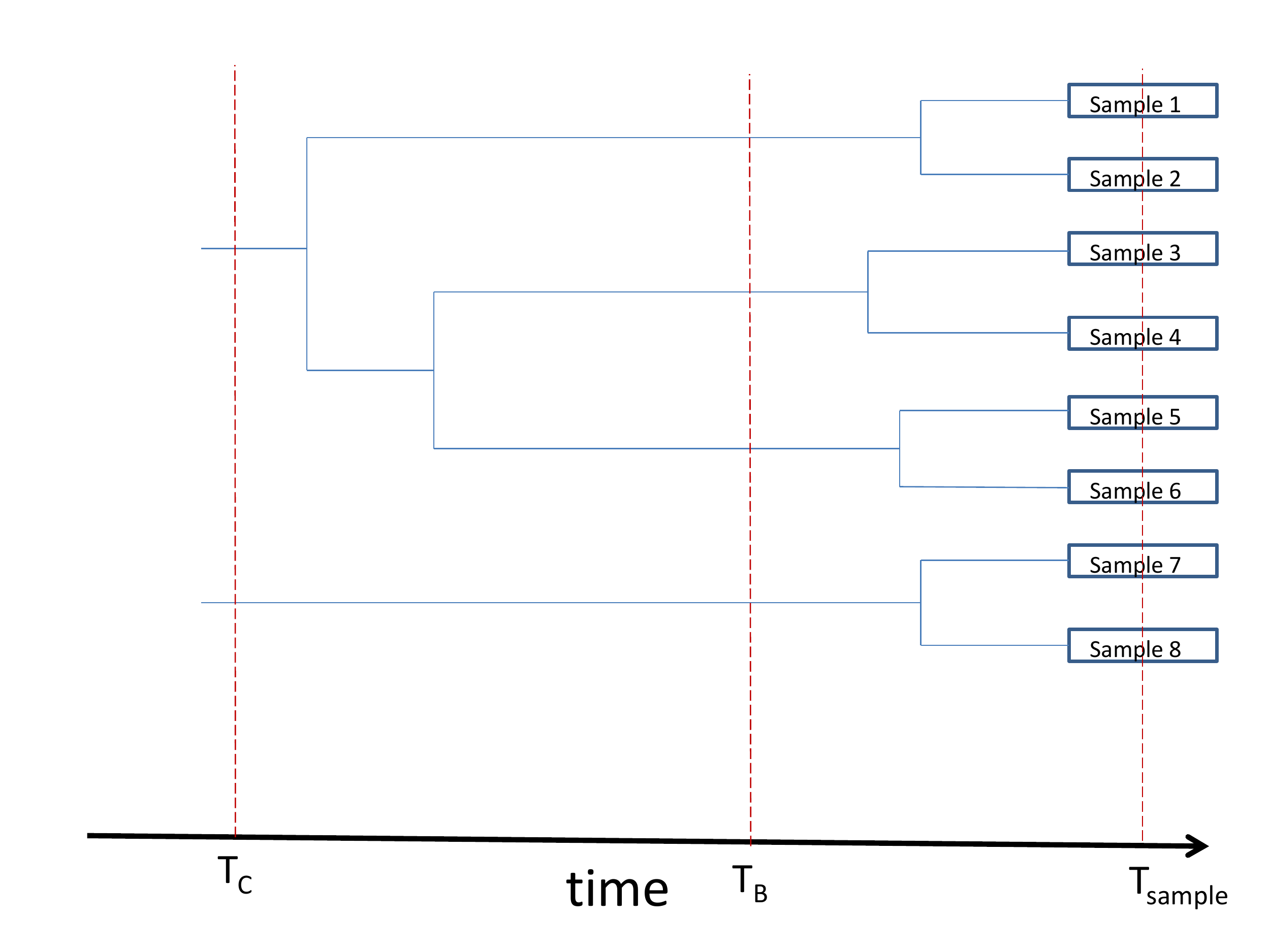}
\caption{Example Lineages.  $n=8$} 
\label{F:partition}
\end{center}
\end{figure}

\section{Results}  \label{S:results}

Our results characterize the lineage structure at $t=0$ of $n$ infected
cells sampled at $t=T_\sample$, a time after HIV has escaped CTL attack.  More
precisely, Theorems \ref{T:linear_theorem} and \ref{T:full_theorem},  
provide a random partition to
which $\Pi(0)$ converges  in the limit $\E \to \infty, \mu \to 0$  with the
limit taken so that $\mu^3 \E^2 \to 0$ and $\mu \E \to \infty$.    We refer to
this limit as the the \textit{small population limit}, \important{\spl}. 
Throughout this work, whenever we take an unspecified limit, we mean the \spl.

In experimental and theoretical HIV studies,  $\E$ has been estimated in the
range $10^6-10^8$.  In  \cite{Chun_1997_Nature}, the number of
activated $\Tcell$ cells with integrated provirus was found to average $3
\times 10^7$.  Various studies have estimated that somewhere between $1$ in
$1000$ to $1$ in $80000$ $\Tcell$ cells are productively infected during
HIV infection, see p. $91$ in \cite{Levy_HIV_Book} and references therein. 
Using a base of $10^{11}$ infectable lymphocytes \cite{Levy_HIV_Book}, this
gives a range of approximately $10^6-10^8$ for $\E$.  Presumably, $\E$ varies
depending on the individual and stage of infection.   Mutation rates for HIV per
base pair per infection event have been estimated at
 approximately $10^{-5}$ \cite{Evolution_HIV_Book}.   Through numerical
 experiments,  we show that Theorems \ref{T:linear_theorem}  and
 \ref{T:full_theorem},  exact in the \spl, are  a good  approximation  for the
 lineage structure formed  under \systems in the parameter regime $\mu =
 10^{-5}$, $\E = 10^6$. In  contrast,  we show that the parameter regime $\mu =
 10^{-5}$, $\E = 10^8$ is not  well approximated by the \spl. The regime $\mu =
 10^{-5}$, $\E = 10^7$ is a middle ground in which the \spls is a reasonable
 approximation, but significant error does exist.   Therefore, we think of the
 \spls as being a limiting version of \systems when  HIV has a relatively small
 infected cell population size.

Theorems \ref{T:linear_theorem} and \ref{T:full_theorem} do not specify the
structure of $\Pi(t)$ at times other than $t=0$.  However, the arguments we use
to justify these theorems do provide some results in this direction which we
mention in the Discussion section.   Similarly, while our
results focus on lineage structure, we make some observations regarding the
stochastic dynamics of \systems
in the Discussion section.   As we mentioned in section \ref{S:ODE}, the $\E \to
\infty$ limit does not eliminate  the stochasticity of (\ref{E:final_system}) in
certain parameter regimes for $\mu$ and the \spls is one such a
regime.

In section \ref{S:theory}, we present
our \spls results, while in section \ref{S:numerical} we discuss the
numerical results that connect the \spls to the parameter regimes of HIV.
   
\subsection{Small Population Limit Results} \label{S:theory}

The dynamics of (\ref{E:final_system}) are 
composed of successive sweeps in which each variant class displaces the  previous variant class as the dominant
portion of the infected cell population.  For example, Figure \ref{F:dynamics}
shows a realization of (\ref{E:final_system}) for a full escape graph with
$\nume=5$, $\delk=.1$, $\gamma=3$, $g=.1$, $\mu=10^{-5}$, $\E = 10^7$.  The
figure was generated by solving \systems numerically.  

\begin{figure} [h]
\begin{center} 
\includegraphics[width=1\textwidth]{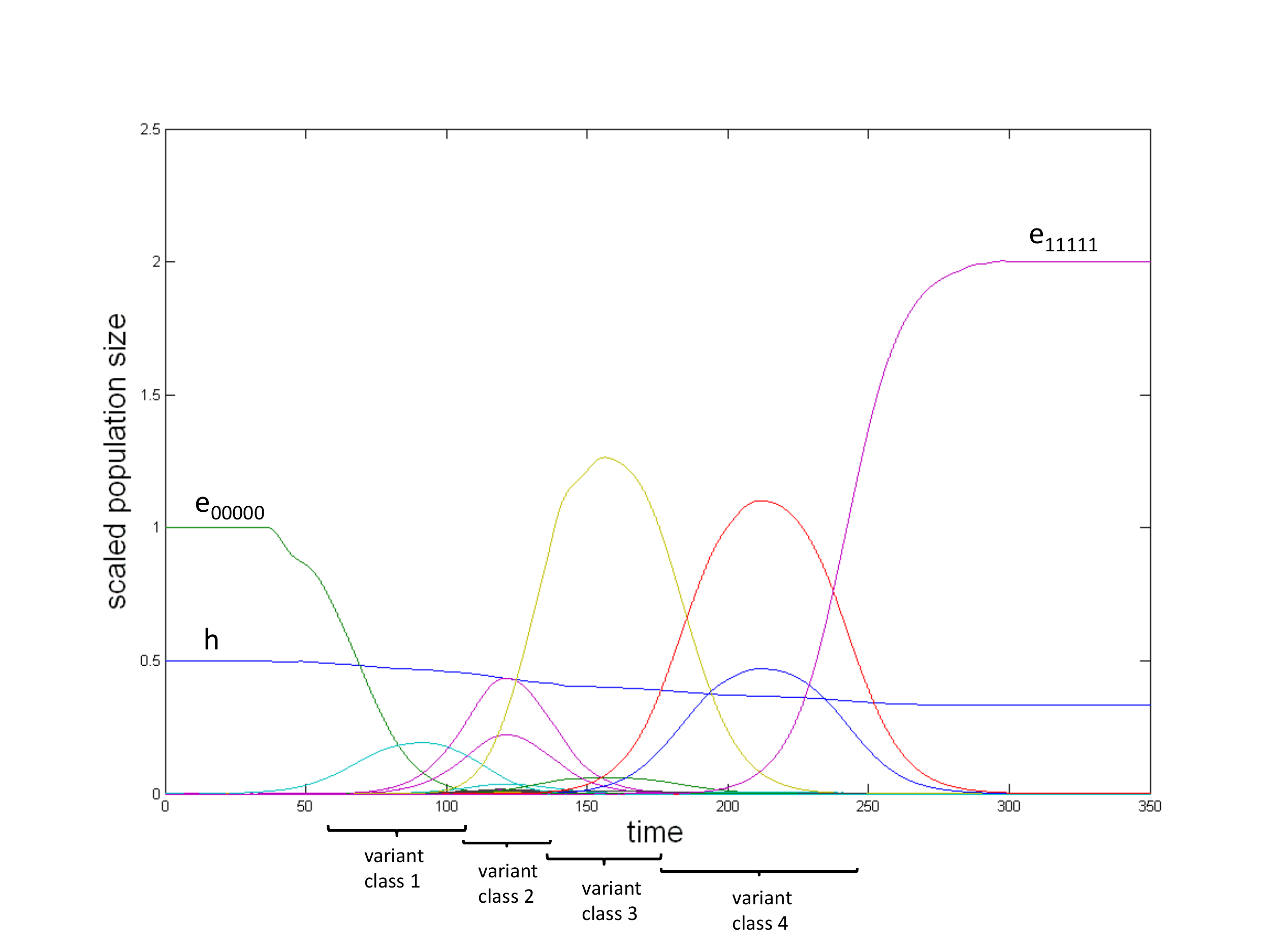}
\caption{Dynamics of (\ref{E:final_system}) for a full escape graph with
$\nume=5$, $\delk = .1$, $\gamma=3$, $\mu = 10^{-5}$, $\E = 10^7$.}
\label{F:dynamics}
\end{center}
\end{figure}

Since initially there are no variants outside of the $0$th and
$1$st variant classes, the variants from  the $i$th variant class with $i>1$
come from $\alle_{i-1} \to \alle_{i}$ mutations, i.e. a mutation $v' \to v$ with
$v' \in \alle_{i-1}$ and $v \in \alle_i$.   The \spls scaling forces such
mutations to occur during a time interval when $\alle_{i-2}$ variants dominate
the population and $\alle_{i-1}, \alle_{i}$ variants are at low frequencies. 
During this time interval, all variants in $\alle_j$ with $j<i-2$ have been
driven out of the population, or nearly so, while all variants in $\alle_j$ for
$j>i$ have yet to arise.   We refer to this time interval as the $\alle_{i-1}$
\textit{spawning phase} because the rise in $\alle_i$ variants is being driven
by $\alle_{i-1} \to \alle_i$ mutations.   At later times, once the $\alle_i$
population has reached higher frequencies, $\alle_{i-1} \to \alle_i$ mutations
have little impact on $\alle_i$ variant population dynamics and the
$\alle_{i-1}$ spawning phase ends.  The condition $\mu^3 \E^2 \to 0$ in the
\spls insures that a variant that is being spawned cannot simultaneously spawn
another variant.

During the $\alle_{i-1}$ spawning phase, variants in $\alle_{i-1}$ are
increasing in population size at approximately rate $\delk$ while variants in
$\alle_i$ are increasing at approximately rate $2\delk$.  To see why,
recall that the $\alle_{i-2}$ variants dominate the infected cell
population during the $\alle_{i-1}$ spawning phase.   Since $\alle_{i-1}$ and
$\alle_i$ variants are attacked by CTLs at $1$ less and $2$ less epitopes
than $\alle_{i-2}$ variants, their relative fitness is given by $\delk$ and
$2\delk$ respectively.   These dynamics are a generalized version of the well
studied Luria-Delbr\"{u}ck (LD) model (see \cite{Zheng_1999_Math_BioSci} for an
excellent review of LD models and results). The  LD model assumes a wild type population growing at rate, say, $a$ that produces mutant types that also grow at rate $a$.  This contrasts to the growth rates $\delk, 2\delk$ for $\alle_{i-1}$ and $\alle_i$ variants respectively in the $i-1$th spawning phase.   For this reason, we refer to spawning phase dynamics as obeying a generalized LD model.  The dynamics of the LD model have  been studied for many years and the LD distribution, which gives
the number  of mutant types at a given time, is well understood.

We analyze \systems by decomposing the time considered, $[0, T_\sample]$, into a
series of time intervals $[T_i, T_{i+1}]$ for $i=0,1,\dots,\nume-2$ along with
intervals $[0,T_0]$, $[T_{\nume-1},T_\nume]$ and $[T_\nume, T_\sample]$. The
interval $[T_{i}, T_{i+1}]$  is the $\alle_{i+1}$
spawning phase.  Intervals $[0,T_0]$, $[T_{\nume-1},T_\nume]$,
$[T_\nume, T_\sample]$ are boundary cases that do not correspond to a
spawning phase.  In this way, we reduce
\systems to a sequence of spawning phases.

For each vertex $v$, we define the \textit{pop value of $v$} as the
number of $v$ variants at the beginning of the $v$ spawning phase.
More precisely, if $v \in \alle_{i+1}$ then the 
pop value of $v$ is $e_v(T_i)$ because $T_i$ is the beginning of $v$'s 
spawning phase.   For the linear escape
graph, we can express the distribution of the $v$th pop value through a simple
formula that is independent of other pop values, see (\ref{E:key}).    In the
case of a full escape graph,  the distribution of the $v$th pop value  is given
by an iterative formula that depends on other pop values, see (\ref{E:full_key}).

Pop values help us get a handle on the stochasticity of \system.  As $\E$
becomes large, the stochasticity of \systems becomes restricted to variants of
small population size.   In our terminology, the stochasticity of \systems
becomes restricted to spawning phases and their corresponding generalized LD
dynamics.   For an interval $[T_i, T_{i+1}]$, pop values describe variant
population sizes at $T_i$ and $T_{i+1}$, thereby giving us a handle on the LD
dynamics that occur between these two times.  

We use extensions of previous LD results to derive our
pop value formulas.   However, these formulas connect to dynamics and we are
interested in forming lineages.  Correspondingly, we need to
understand not only the dynamics of the LD model  but also how to construct
lineages on a generalized LD model.   The random partition
$\Xi_{A,i}$, which we discuss more thoroughly below, characterizes the 
coalescent events on  lineages as they move backwards in time through
generalized  LD dynamics  corresponding to a single spawning phase. To form
lineages, we  consider  a sequence of spawning phases. For a linear escape
graph, this is done  through simple concatenation of $\Xi_{A,i}$.   But for the
full escape  graph,  things are more complicated as variants within a variant
class affect  each others dynamics and hence each others lineages.
    
\subsubsection{linear escape graph}

As mentioned, constructing lineages for linear escape
graphs is just a matter of concatenating the coalescent events associated with
each spawning phase.   We characterize such coalescent events through the random
partition $\Xi_{A,i}$ which we now define.

For $i=0,1,\dots, \nume-2$ we define a r.v. $\Gamma^{(i)}$ by,
\begin{equation}
\Gamma^{(i)} = \exp[2*W_1] W_2 B(\frac{2 \delk}{k_{i}}).
\end{equation}
where $W_1, W_2$ are independent exponential r.v's with mean $1$ and $B(p)$
is an independent Bernoulli r.v. with success probability $p$.  

We define $\Xi_{A,i}$ through a paintbox construction as follows (see
\cite{Kingman_1978_Proc_R_Soc_Lond_A, Pitman_2002_St_Flour_Lecture_Notes} for a
review of paintbox constructions).
\begin{definition}  \label{D:linear_xi}
Let $A>0$ and $i \in \{0,1,\dots, \nume-2\}$ be given.   Then we define a
partition $\Xi_{A,i}(S)$ on any set $S$ as follows.   Let $K$ be a
sample from a Poisson r.v. with mean $A$.    For $j=1,2,\dots,K$  we sample
$\Gamma_{j}$ from the r.v. $\Gamma^{(i)}$.

Thinking of the $j$ as colors, we 'paint' each $s \in S$ a random color
according to the probability
\begin{equation}  \label{E:linear_paint}
P(\text{paint with color } j) = \frac{\Gamma_{j}}
				{\sum_{j'=1}^{K} \Gamma_{j'}}
\end{equation}
The partition $\Xi_{A,i}(S)$ is formed by grouping together elements
sharing the same color.
\end{definition}
The $K$ samples of $\Gamma^{(i)}$ correspond to $K$, $\alle_{i+1} \to
\alle_{i+2}$ mutations on the $\alle_{i+1}$ spawning phase, $[T_i, T_{i+1}]$. 
Each $\Gamma^{(i)}$ sample  is proportional to the number of descendants at
$T_{i+1}$ produced by a single such mutation.  
Essentially, the $\Gamma^{(i)}$ samples provide a decomposition for
the pop value of $v \in \alle_{i+2}$.   More precisely,
\begin{equation}  \label{E:pop_decomp_linear}
e_v(T_{i+1}) = C \sum_{j=1}^K \Gamma_j,
\end{equation}
where $C$ is a constant independent of $j$.   The decomposition
(\ref{E:pop_decomp_linear}) allows us to construct lineages, while simply
sampling the pop value would not. 
Roughly, the number of infected cells at $T_{i+1}$ that descend from the 
$j$th $\alle_{i+1} \to
\alle_{i+2}$ mutation is proportional, with constant $C$ in
(\ref{E:pop_decomp_linear}), to $\Gamma_j$.  This means that a sampled cell will
descend from mutation $j$ with probability given by (\ref{E:linear_paint}).
Sampled cells that descend from the same mutation on $[T_i, T_{i+1}]$ must
coalesce during that period.  In this way $\Xi_{A,i}$ characterizes coalescent
events on $[T_i, T_{i+1}]$.

The parameter $A$ is a tuning parameter. As Theorem
\ref{T:linear_theorem} shows, raising $A$ improves accuracy by considering more
mutations, but at a computational cost of increasing the number of
$\Gamma^{(i)}$ samples that must be taken.

For the linear escape graph, $\Pi(0)$ is simply a concatenation of the
$\Xi_{A,i}$.

\begin{claim}  \label{T:linear_theorem}
Consider (\ref{E:final_system}) assuming a linear escape graph.    Then letting
$\Delta$ be any partition of the $n$ lineages,
\begin{equation}
\lim P(\Pi(0) = \Delta) = P(\left(\prod_{i=0}^{\nume-2} \Xi_{A,
i}\right)(\Pi(T_\sample)) = \Delta) + O(\frac{1}{A})
\end{equation}
where $\Xi_{A,i}$ is given by definition \ref{D:linear_xi}.  
\end{claim}
Recall that $\Pi(T_\sample)$ simply partitions each lineage separately since
no coalescent events have occurred.  By $\left(\prod_{i=0}^{\nume-2} \Xi_{A,
i}\right)(\Pi(T_\sample))$ we mean the concatenation of the $\Xi_{A,i}$ applied
to $\Pi(T_\sample)$.  For example if $\nume=3$ then,
\begin{equation}
\left(\prod_{i=0}^1 \Xi_{A,
i}\right)(\Pi(T_\sample)) = \Xi_{A,0}(\Xi_{A,1}(\Pi(T_\sample))).
\end{equation}

\subsubsection{full escape graph}

As mentioned, the $i+1$th spawning phase involves $\alle_{i+1} \to \alle_{i+2}$
mutations. For the linear escape graph,  there is only one variant in each
variant class, meaning that there is only one type of
$\alle_{i+1} \to \alle_{i+2}$ mutation.   However, for the full escape graph we
must consider $v' \to v$ mutations for every $v \in \alle_{i+2}$ and $v' \in
\P(v)$. The time period $[T_{i}, T_{i+1}]$ will be composed of many
concurrent spawning phases, one for each such $v' \to v$.

To explain the generalization of Theorem \ref{T:linear_theorem} to the full
escape graph, recall that for the linear escape graph $\Xi_{A,i}$ is formed by
taking $K$ samples of $\Gamma^{(i)}$ and $K$ is always sampled from a Poisson
r.v. with mean $A$. In the full escape graph case, for each $v \in \alle_{i+2},
v' \in \P(v)$ we take $K_{v' \to v}$ samples of $\Gamma^{(i)}$,  where  $K_{v'
\to v}$  is sampled from a Poisson r.v.  with mean that depends on the pop value
of $v'$  relative to the other vertices in the $\alle_{i+1}$ class. 

To explain why $K_{v' \to v}$ should depend on pop values, consider $v',v'' \in
\P(v)$.  Suppose $e_{v'}(T_i) \gg e_{v''}(T_i)$.  In other words, $v'$ has a much
higher pop value than $v''$.   A higher pop value will mean that on $[T_i,
T_{i+1}]$, more $v' \to v$ mutations occur then $v'' \to v$
mutations and correspondingly we should have $K_{v' \to v} > K_{v'' \to v}$.  
This effect did not arise in the linear escape graph because each variant class
contains a single variant.  

Since $K_{v' \to v}$ depends on pop values, intuitively we must first sample pop
values and then construct the $K_{v' \to v}$.  However, as in (\ref{E:pop_decomp_linear}), 
to form lineages we do not sample pop values.  Rather we
decompose each pop value according to the number of descendants produced
by each $v' \to v$ mutation.   The decomposition depends on $K_{v' \to v}$. 
Putting these comments together,  we must build $K_{v' \to v}$ and
pop value decompositions simultaneously.   This is done in Definition
\ref{D:pop}. The variable $D_v$ is proportional to the pop value of $v$
and is formed through a decomposition analogous to (\ref{E:pop_decomp_linear}).

\begin{definition} \label{D:pop}
For each $v \in \alle_1$ we define $D_v = 1$.
Then we recursively define $D_v$, $K_{v' \to v}$ and $\Gamma^{(i)}$ samples
 as follows. Suppose the $D_v$ values are known
for $v \in \alle_{i-1}$. Set
\begin{equation}
D_{\max,i-1} = \max_{v \in \alle_{i-1}} D_v
\end{equation}
For each $v \in \alle_i$ and $v' \in \P(v)$ we let $K_{v' \to v}$ be a
sample from a Poisson r.v. with mean $A(D_{v'}/D_{\max,i-1})$.   For each
$j=1,2,\dots,K_{v' \to v}$, we sample $\Gamma_{v' \to v, j}$ from
$\Gamma^{(i)}$.  Then, 
\begin{equation}
D_v = \sum_{v' \in \P(v)} \sum_{j=1}^{K_{v' \to v}} \Gamma_{v' \to v, j}
\end{equation}
\end{definition}
The above definition allows us to define $K_{v' \to v}$ and $\Gamma$ samples for
every mutation pair $v' \to v$.  The pop value of $v \in \alle_{i+1}$ is given
by,
\begin{equation} \label{E:pop_decomp_full}
e_v(T_i) = C_i D_v = C_i \sum_{v' \in \P(v)} \sum_{j=1}^{K_{v' \to v}}
\Gamma_{v' \to v, j},
\end{equation}
where $C_i$ depends only on $i$.  
(\ref{E:pop_decomp_full}) is analogous to (\ref{E:pop_decomp_linear}). 

For the full escape graph, the state of our
lineages is not simply a partition of $\{\lin_1,\lin_2,\dots,\lin_n\}$.  
Rather,  we must specify a vertex to
which each lineage is associated at a given time $t$.  Intuitively, the vertex
associated with, say, $\lin_j$ at time $t$ is the variant type of the infected
cell at time $t$ from which the $j$th sampled cell descends.    To put this in the
context of a partition function, $\Pi(t)$ partitions the lineages into  disjoint
sets and associates with each such set a vertex in $\alle$.   The $\Xi_{A,i}$
defined below are random partitions on sets for which every element is
associated with a vertex in $\alle_{i+2}$.
 
\begin{definition}  \label{D:full_xi}
We define a partition $\Xi_{A,i}(S)$ on a set $S$ for which each element $s \in
S$ is associated with a vertex $v_s \in \alle_{i+2}$.   For every $s,v_s$ pair
and $v' \in \P(v_s)$ let $K_{v' \to v_s}$, $D_{v_s}$ and associated $\Gamma_{v'
\to v_s,j}$ be as defined in Definition \ref{D:pop}.  Assign a unique color to every triple
$(v',v_s,j)$.  Then we paint each element  $s \in S$ the color associated with
$(v',v_s,j)$ and assign it element $v'$ with the
following probabilities,
\begin{equation}
P(\text{paint with color } (v',v_s,j) \text{ and assign vertex }v') =
\frac{\Gamma_{v' \to v_s, j}} {D_{v_s}}
\end{equation}
The partition $\Xi_{A,i}(S)$ is formed by grouping together elements
sharing the same color.
\end{definition}

With the adjusted definition of $\Xi_{A,i}$, the statement of Theorem
\ref{T:linear_theorem} now holds for the full escape graph.  

\begin{claim}  \label{T:full_theorem}
Consider (\ref{E:final_system}) assuming a full escape graph.    Then letting
$\Delta$ be any partition of the $n$ samples,
\begin{equation}
\lim_\spl P(\Pi(0) = \Delta) = P(\left(\prod_{j=0}^{\nume-2} \Xi_{A,
i}\right)(\Pi(T_\sample)) = \Delta) + O(\frac{1}{A})
\end{equation}
where $\Xi_{A,i}$ is given by definition \ref{D:full_xi}
\end{claim}
For the full escape graph, $\Pi(T_\sample)$ partitions each lineage separately
and assigns to each lineage the vertex $\tb{11\dots1}$ since we sample after HIV has escape CTL attack. 
$\Delta$ should assign to each lineage a
variant of class $\alle_0$ or $\alle_1$ since these are the only variants extant
at $t=0$.  However, for simplicity Theorem \ref{T:full_theorem} refers to
the partition structure of the lineages at  $t=0$.  (If we wanted to include the
variant associated with each lineage at $t=0$, we would need a
$\Xi_{A,-1}$.  Our methods allow for this, but for the sake of simplicity and
because our initial conditions are slightly ad-hoc, we do not consider such an
extension.)

\subsection{Numerical Results} \label{S:numerical}

In this section we consider five parameter regimes:  the
approximating regime (AR), the small population
regime (SPR), the medium population regime (MPR), and the large population
 regime (LPR).  Table \ref{T:parameter} specifies  the $\mu$ and
$\E$ value associated with each regime.    The table also
includes the corresponding values for $\mu^3\E^2$ and $\mu \E$.    The AR 
has $\mu^3\E^2 \ll 1$ while  $\mu \E \gg 1$, suggesting a good approximation by
the $\spl$.    Notice that the SPR has a scaling near the $\spl$, but that the
LPR has a $\mu^3 \E^2$ value of $10$ which, as we shall  show, is too large for
the $\spl$ to apply.

\begin{table}
\begin{center}  
\begin{tabular}{|c||c|c|c|c|} 
\hline 
regime & $\mu$ & $\E$ & $\mu^3 \E^2$ & $\mu \E$ \\
\hline
\spl & - & - & $0$ & $\infty$ \\
\hline
AR & $10^{-10}$ & $10^{13}$ & $.0001$ & $1000$ \\
\hline
SPR & $10^{-5}$ & $10^{6}$ & $.001$ & $10$ \\
\hline
MPR & $10^{-5}$ & $10^{7}$ & $.1$ & $100$ \\
\hline
LPR & $10^{-5}$ & $10^8$ & $10$ & $1000$\\
\hline
\end{tabular}
\end{center}
\caption{Parameter Regimes}
\label{T:parameter}
\end{table}

To understand the accuracy of the \spl, we first consider the probability that
two sampled lineages coalesce.   More precisely, setting $n=2$ we
consider the probability that $\lin_1$ and $\lin_2$ coalesce by $t=0$ or
equivalently $P(\Pi(0) = \{\{\lin_1, \lin_2\}\})$.  This probability is often
computed in population genetics applications and is one way to characterize
$\Pi(0)$ \cite{Wakeley_Book_Coalescent}.   Figure \ref{F:linear_coal_prob} shows
this coalescent probability for a  linear escape graph with $\gamma = 3$ and
$\delk = .1$.    The x-axis gives the number of epitopes in the linear 
escape graph, our parameter $\nume$.  We skip $\nume=1$ because due to our
initial conditions such an attack has a coalescent probability near zero.   For
each value of $\nume$ considered, we computed five quantities given by the five
bars.  The left most bar represents the coalescent probability given in the
\spl, as specified  through Theorem \ref{T:linear_theorem}.   To compute the
coalescent probability  in this case, we set $A=100$, higher values of $A$ don't
change the result,  and constructed $\Xi_{A,i}$ for $i=0,1,\dots,\nume-2$. This
amounts to sampling  the r.v.'s $\Gamma^{(i)}$.    We generated $1000$
realizations of the  sequence of $\Xi_{A,i}$.    For each realization, we then
applied the paintbox  construction implied by the underlying $\Gamma^{(i)}$
samples to determine if the two lineages coalesced.  We did this $1000$ times 
for each realization of the $\Xi_{A,i}$.  In this way we computed one million
$1$'s, for coalescence, and $0$'s, for non-coalescence.   Averaging this list
gave us the coalescent probability.

\begin{figure} [h]
\begin{center} 
\includegraphics[width=1\textwidth]{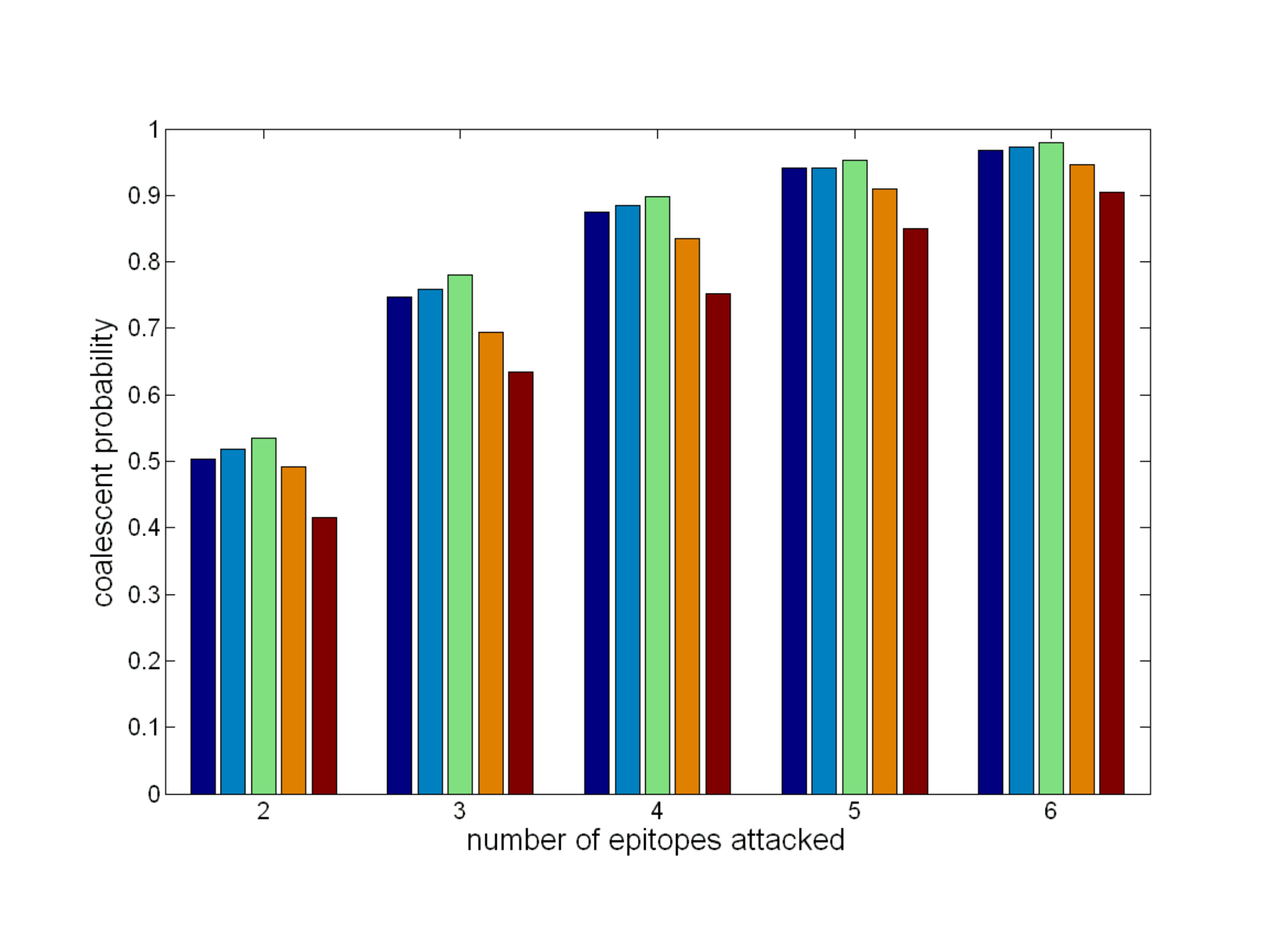}
\caption{The probability of coalescence of two lineages for
a linear escape graph with $\gamma = 3$, $g=.1$, $\delk=.1$.   The bars, from
left to right, give the coalescent probability under the \spl, AR, SPR, MPR,
and LPR (see Table \ref{T:parameter} for the definition of these parameter regimes).}
\label{F:linear_coal_prob}
\end{center}
\end{figure}

The next four bars represent, from left to right, the coalescent
probability for the AR, SPR, MPR and LPR, respectively.   These
values are computed by solving \systems numerically and forming lineages on top
of the stochastic dynamics.  We compute $1000$ realizations of \systems
dynamics, and for each such realization we consider the coalescence of $2$
lineages $1000$ times.   We then average over all $1000$ lineage pairs and all
$1000$ realizations. Solving \systems and building lineages on top of the
dynamics is not numerically trivial due to the large population size.  
Following methods described in \cite{Leviyang_2011_BMB}, we solve \systems
exactly and track parent-child relationships in each variant  until the variant population size
exceeds $10000$.  At that point we switch to the deterministic ODE analogue of
\system.   

As Figure \ref{F:linear_coal_prob} demonstrates, the \spls is a good
approximation in the AR and SPR, but not the LPR. The MPR represents a middle
ground. Figure \ref{F:full_coal_prob} is the same as Figure
\ref{F:linear_coal_prob}, except that in Figure \ref{F:full_coal_prob}, we consider a full escape graph.  

\begin{figure} [h]
\begin{center} 
\includegraphics[width=1\textwidth]{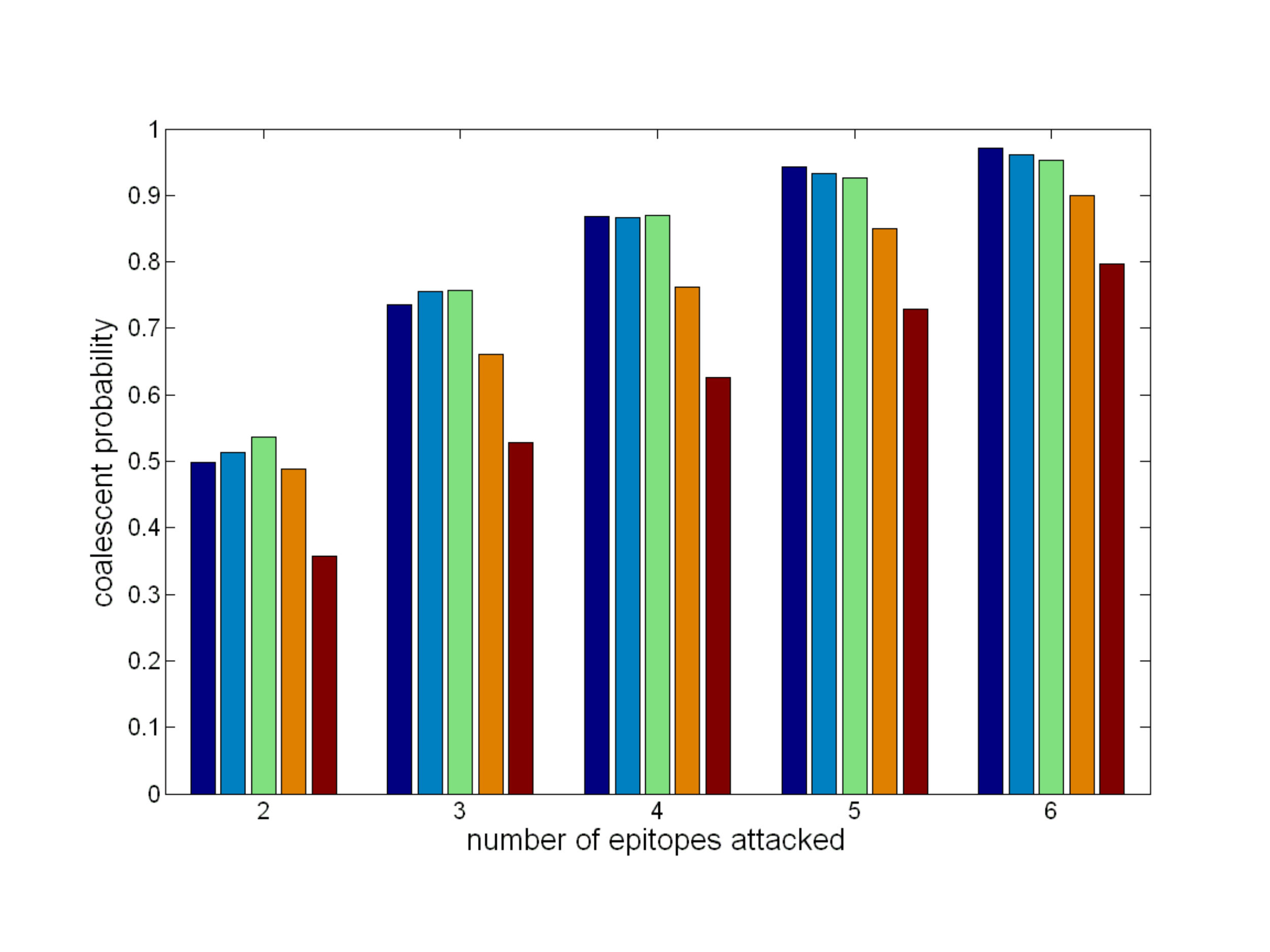}
\caption{The probability of coalescence of two lineages for 
a full escape graph with $\gamma = 3$, $g=.1$, and $\delk=.1$.   The bars, from
left to right, give the coalescent probability under the \spl, AR, SPR,
MPR, and LPR (see Table \ref{T:parameter} for the definition of these parameter
regimes).}
\label{F:full_coal_prob}
\end{center}
\end{figure}

Another value that can be used to characterize the HIV genealogy shaped by CTL
attack is the number of still uncoalesced lineages at $t=0$.  More
precisely, we consider the number of elements  in $\Pi(0)$.  Recall that each
element of $\Pi(0)$ is a collection of lineages that have coalesced.  Figure
\ref{F:full_num_coal} shows the distribution of this number for a full escape
graph with $\nume = 3$, $\gamma = 10$, $\delk = .3$ and $n=100$.  The same
pattern of accuracy is seen as with Figures \ref{F:linear_coal_prob} and 
\ref{F:full_coal_prob}.

\begin{figure} [h]
\begin{center} 
\includegraphics[width=1\textwidth]{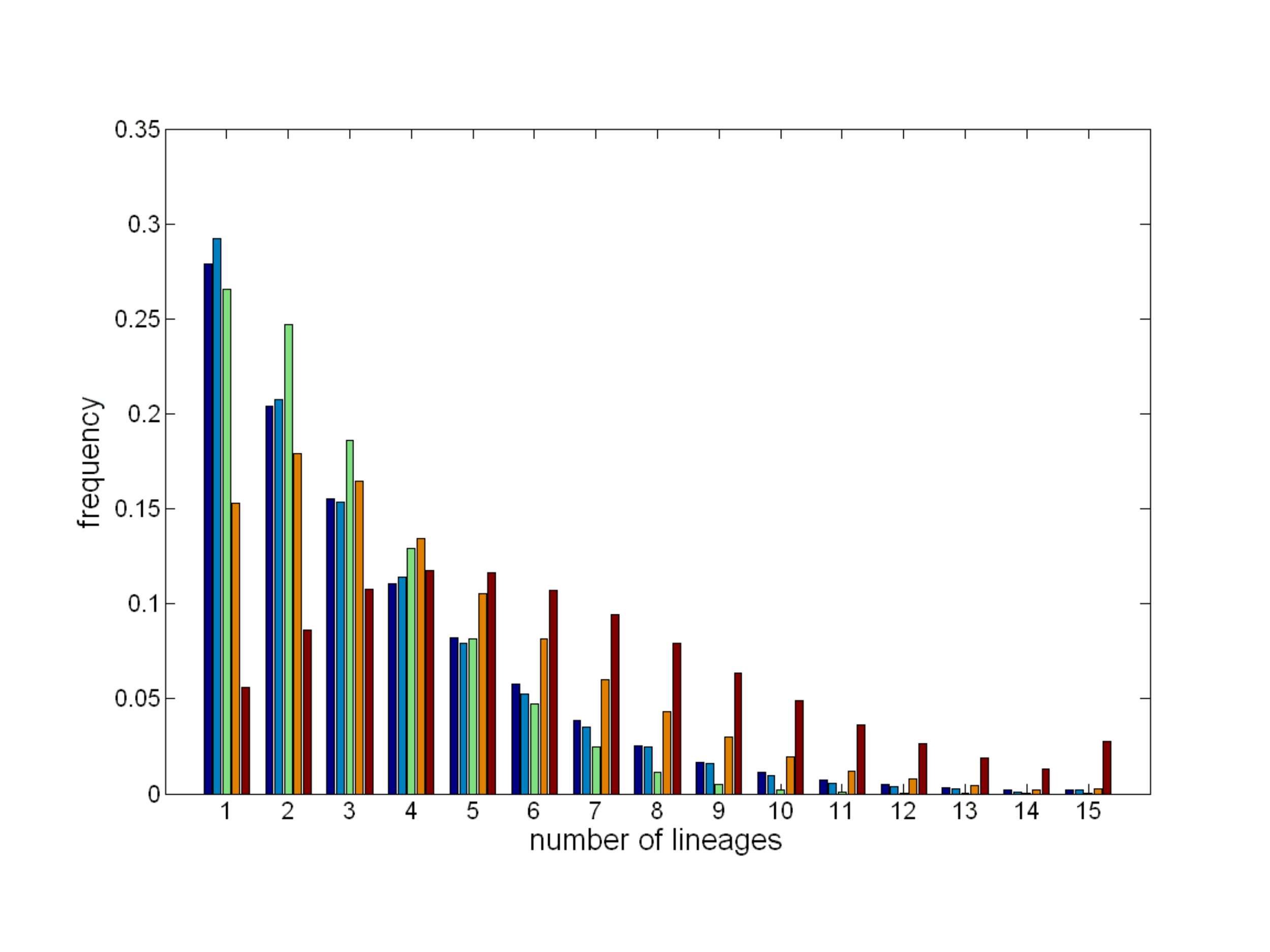}
\caption{The number of uncoalesced lineages at $t=0$ assuming a sample
of $100$ infected cells after HIV escape.   Results correspond to a 
full escape graph with $\nume=3$, $\gamma=10$ and $\delk = .3$.
The bars, from left to
right give values under the \spl, AR, SPR, MPR, and LPR (see
Table \ref{T:parameter} for the definition of these parameter regimes).}
\label{F:full_num_coal}
\end{center}
\end{figure}

Theorems \ref{T:linear_theorem} and \ref{T:full_theorem} provide a theoretical
framework for understanding genealogies on \system.  However, they also provide
a computational approach for sampling such genealogies
that is much faster than solving \systems directly.    Table \ref{T:speed} gives
the CPU time in seconds required to generate the coalescent probability results
shown in Figures \ref{F:linear_coal_prob} and \ref{F:full_coal_prob} for the cases
$\nume=2,6$. We show CPU times needed to produce the probability through our
\spls results and by solving \systems in the SPR,  the times required for
the APR, MPR, and LPR are similar to the SPR.     As can be seen,
the \spls approach is more than $300$ times faster  in the case of a full escape
graph and $\nume=6$.   For the linear escape  graph and the $\nume=2$ full
escape graph, the \spls approach is on the order of $100$ times faster.

\begin{table}[t]
\begin{center}
\begin{tabular}{|c|c||c|c|c|}  
\hline 
graph  & $\nume$ & $\spl$ time & SPR time & SPR/\spl \\
\hline
LINEAR & 2 & 29 & 4700 & 162\\
\hline
& 6 & 135 & 11300 & 84\\
\hline
FULL & 2 & 32 & 5300 & 165\\
\hline
& 6 & 159 & 54000 & 340\\
\hline
\end{tabular}
\caption{CPU Time in seconds needed to generate coalescent probabilities.  CPU
is an Intel-Celeron single node processor.}
\label{T:speed}
\end{center}
\end{table}

\section{Discussion}  \label{S:discussion}

Application of the results we have presented depends on approximating the \spls
scaling by satisfying $\mu \E \gg 1$ and $\mu^3 \E^2 \ll 1$.   HIV almost
certainly always satisfies $\mu \E \gg 1$, so this condition is not
restrictive.   On the other hand, $\mu^3 \E^2 \ll 1$
is satisfied if the HIV population size is of relatively small magnitude, 
namely on the order of $10^6$.   Importantly, the HIV population size
we must consider is the number of active $\Tcell$ cells infected with
functioning HIV genome. Inactivated $\Tcell$ cells or those infected by
non-functional HIV do not enter into our model because they do not produce offspring infected
cells.  

Our results have implications for both dynamics
and genealogies.   For dynamics, our arguments show that the
stochasticity   of \systems in the \spls is completely contained within the pop
values described in the results  section and defined precisely in (\ref{E:key}) and
(\ref{E:full_key}).   Intuitively, once a variant population reaches large size,
averaging effects take  over and deterministic dynamics apply.  In our
nomenclature, a variant  population is small and hence experiences stochastic
dynamics only when it is  being spawned by another variant population.   These
spawning dynamics are  encoded in the pop values which are stochastic.  Taking
all this together, if we  are interested in dynamics and not lineages, then
\systems can be reduced to a deterministic ODE accompanied by stochastic pop values.   

The stochasticity of the pop values has significant impact on the dynamics of
\system.   Figure \ref{F:dynamics_det} gives a solution for the deterministic
analogue of \systems in which stochastic events are replaced by their average. 
Another way to describe such a system is as \systems when all pop values are
equal.   Either way, since our equations are symmetric, the dynamics must be
symmetric and this is indeed the case in Figure \ref{F:dynamics_det}.  All
variants within the same variant class have identical dynamics.   

\begin{figure} [h]
\begin{center} 
\includegraphics[width=1\textwidth]{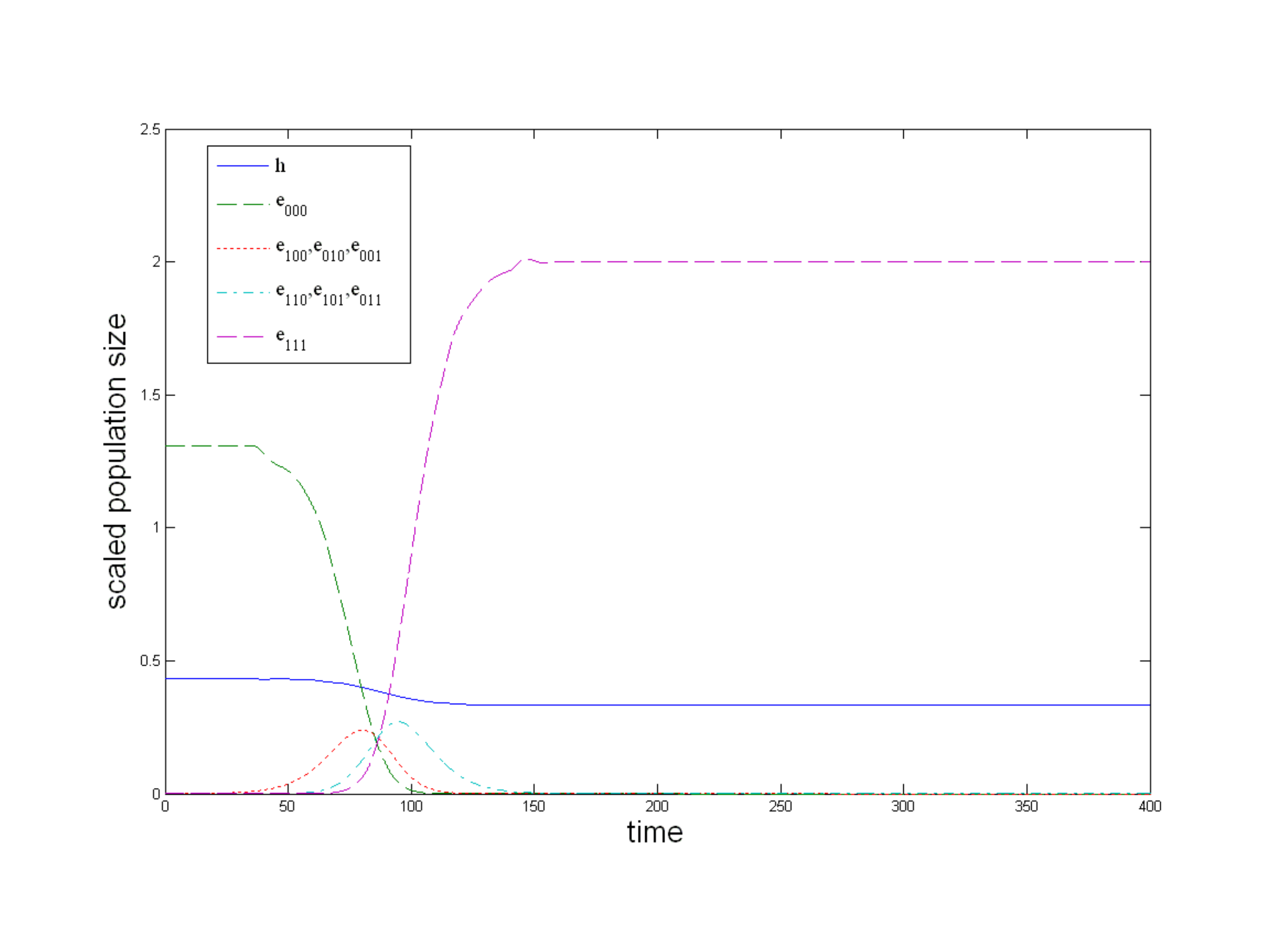}
\caption{\systems run deterministically for a full escape graph with $\nume=3$,
$\delk = .1$, $\gamma = 3$, $\mu = 10^{-5}$, $\E = 10^6$.  As a consequence of
symmetry, all variants in the same epitope class evolve identically.}
\label{F:dynamics_det}
\end{center}
\end{figure}

In contrast, Figure \ref{F:dynamics_stoch} provides the dynamics for a single
realization of \system.  We can see that stochasticity plays an essential role
in \systems because Figure \ref{F:dynamics_stoch} gives very different dynamics
than Figure \ref{F:dynamics_det}.  But further, our work explains the
stochasticity seen in Figure \ref{F:dynamics_stoch}.   In a given variant class,
some variants have higher pop values than others. Such variants dominate the
others in their class.  For example in Figure \ref{F:dynamics_stoch}, the
variant $\textbf{011}$ dominates $\textbf{110, 101}$ when these variants compose
most of the population.    This dominance results from stochasticity
corresponding to a high pop value.  Biologically,  the stochasticity of pop
values come from the stochasticity of mutation times.

\begin{figure} [h]
\begin{center}  
\includegraphics[width=1\textwidth]{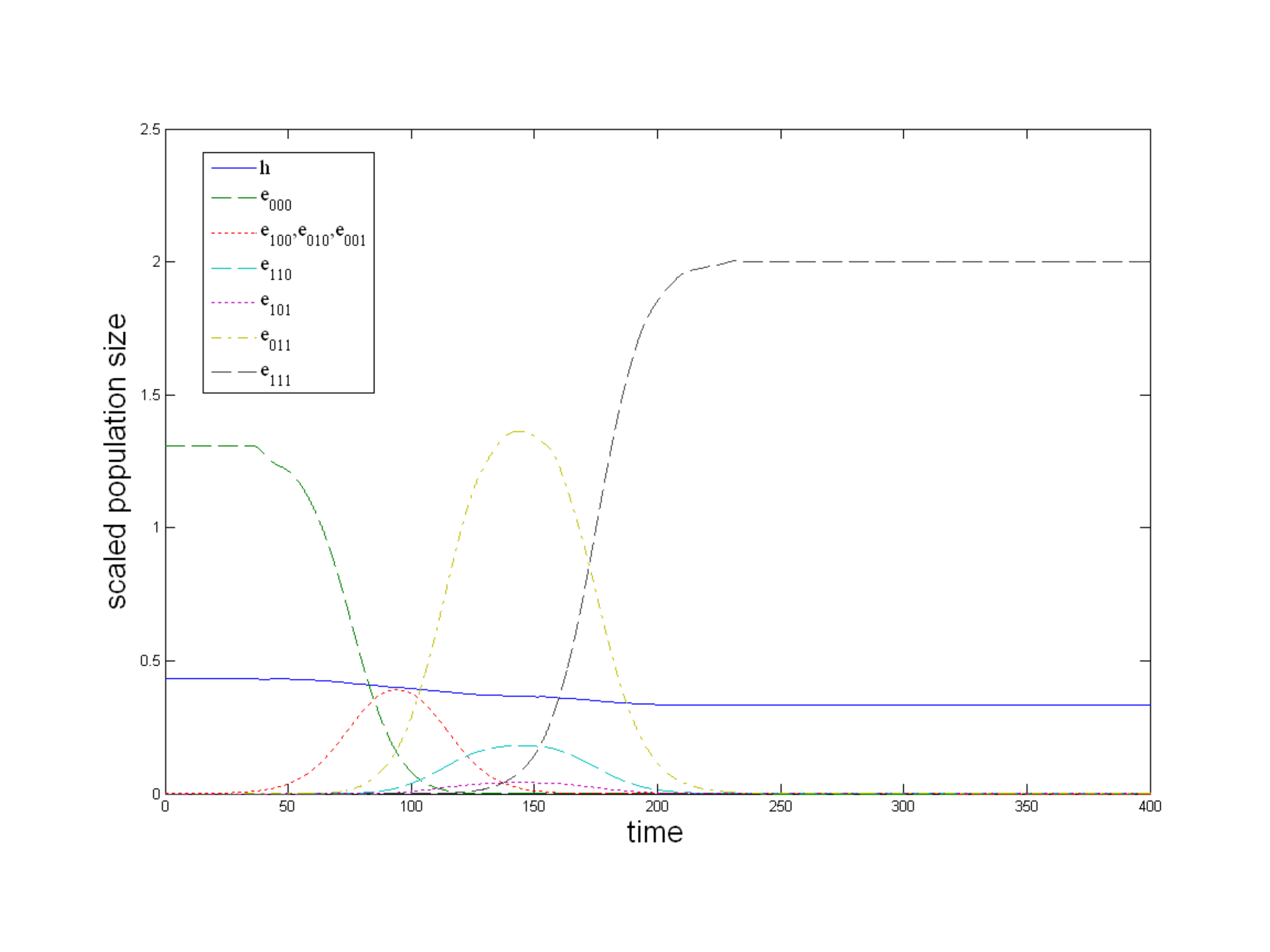}
\caption{\systems run for a full escape graph with $\nume=3$,
$\delk = .1$, $\gamma = 3$, $\mu = 10^{-5}$, $\E = 10^6$.  Unlike the
deterministic case shown in Figure \ref{F:dynamics_det}, the symmetry of the
model is broken by stochastic effects.}
\label{F:dynamics_stoch}
\end{center}
\end{figure}

Turning now to genealogies, we have described the coalescence of lineages
caused by the whole period of HIV escape.   However, as mentioned, we can
decompose HIV escape into time intervals  $[T_{i}, T_{i+1}]$.   Each such period corresponds to
$\Xi_{A, i}$ and so we know the  state of the lineages at each $T_i$ given
the state at $T_{i+1}$. Between the $T_i$, however, our results  do not describe
the lineages. Figures \ref{F:gen_5} and \ref{F:gen_2} show  genealogies formed
for a $5$ epitope and $2$ epitope attack, respectively, in  the case of a linear escape graph under
the SPR.    Here we have shown all coalescent
events that happen during $[T_{i}, T_{i+1}]$ to occur at
$T_{i}$.  Both genealogy figures were produced using Figtree.  (Figtree is
available as part of the BEAST software package \cite{BEAST}.)

\begin{figure} [h]
\begin{center} 
\includegraphics[width=1\textwidth]{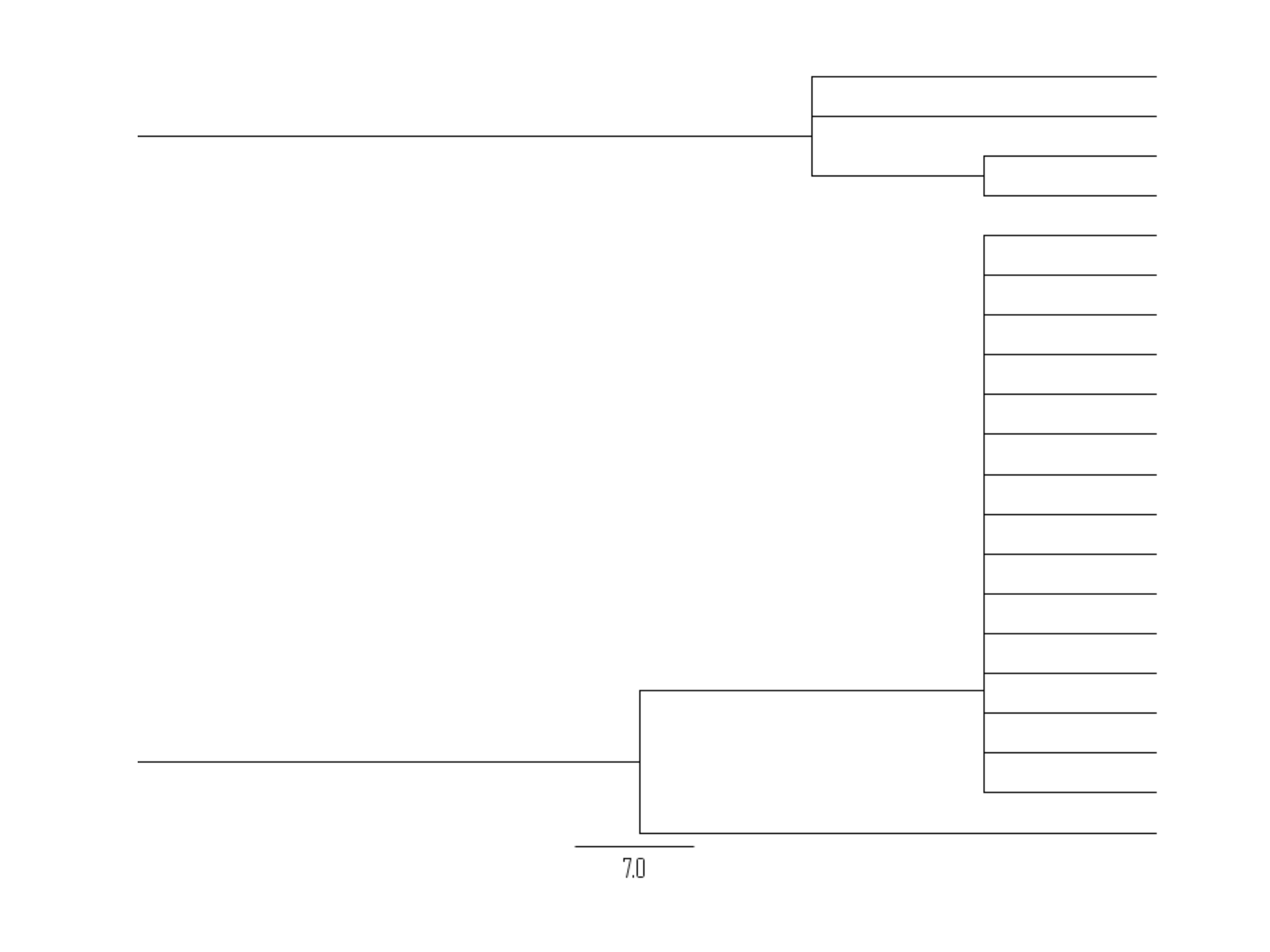}
\caption{Sampled genealogy for a linear escape graph with $5$ epitopes
attacked under the SPR, that is $\mu=10^{-5}, \E=10^6$.   $\gamma = 3$, $g=.1$,
and $\delk=.1$ as in Figure \ref{F:linear_coal_prob}.  
$n=20$.   The time scale at the bottom is in units of $2$ days.}
\label{F:gen_5}
\end{center}
\end{figure}

\begin{figure} [h]
\begin{center} 
\includegraphics[width=1\textwidth]{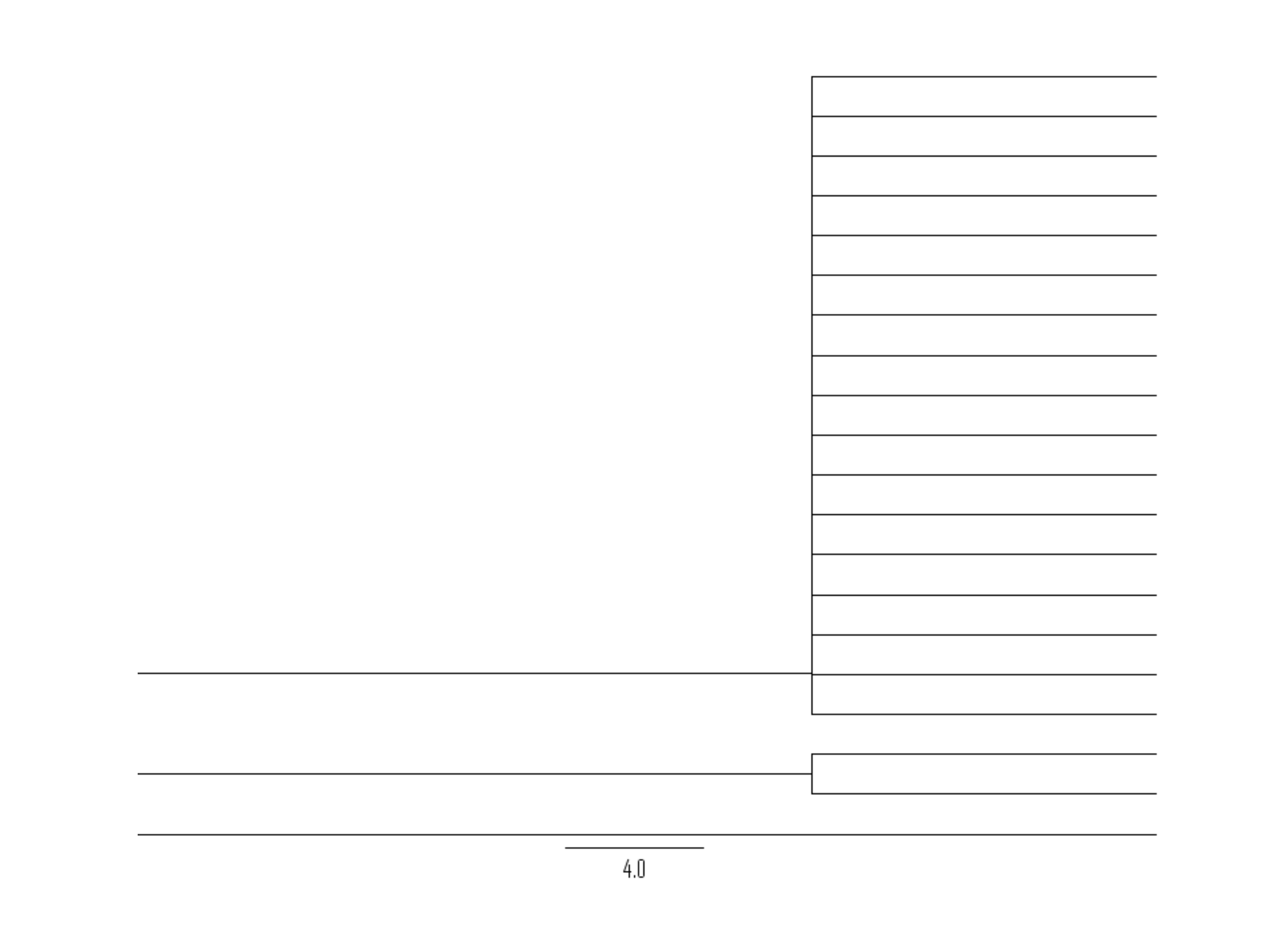}
\caption{Sampled genealogy for a linear escape graph with $2$ epitopes
attacked under the SPR, that is $\mu=10^{-5}, \E=10^6$.  $\gamma = 3$, $g=.1$,
and $\delk=.1$ as in Figure \ref{F:linear_coal_prob}.  
$n=20$.   The time scale at the bottom is in units of $2$ days.}
\label{F:gen_2}
\end{center}
\end{figure}

The restriction of our current results to symmetric attack and the small end of
the HIV population size range is a significant limitation.  Further work
should allow for these restriction to be lifted, but our current results
provide some general observations.   

For a full escape graph, the assumptions of symmetric attack makes the paths
through the graph identical in terms of the underlying parameters.    Removal of
the symmetric attack assumption would lead to a dominant path.    For example,
if there is an epitope that is attacked more strongly than  the  other epitopes,
then it will  be the first epitope at which HIV escapes CTL attack.    Of
course, there will be some HIV variants that initially  posses a mutation at  a
different epitope, but these will be few in number.    The order of the epitopes
at which HIV escapes from the CTL attack will be specified in  the
case of asymmetric CTL attack.   As a result,  HIV escape on a full escape 
graph in the asymmetric attack case  should proceed essentially on one path of
the graph and be similar  to the linear escape graph dynamics  and genealogies we have
discussed.

Our numeric results allow us to compare the form of genealogies for large and
small HIV populations.    Figure \ref{F:compare_coal_prob} compares coalescent
probabilities for linear and full escape graphs.   This is the same data
presented in Figures \ref{F:linear_coal_prob} and \ref{F:full_coal_prob}.    In
Figure \ref{F:compare_coal_prob} the  four bars give, from left to right, the
coalescent probability for a linear  escape graph under SPR, a full escape graph
under SPR, a linear escape graph  under LPR, and a full escape graph under LPR. 
As can be seen, the coalescent  probabilities under the SPR are similar
for the linear and full escape graphs.  Some numerical experiments suggest that
this is because pop value stochasticity causes a single path through the full
escape graph to dominate, similarly to our earlier comments on asymmetric
attack.   We don't know why this is not the case for the LPR.  It may be that
pop values take on a different form in this regime to which our \spls analysis
does not apply.

\begin{figure} [h]
\begin{center} 
\includegraphics[width=1\textwidth]{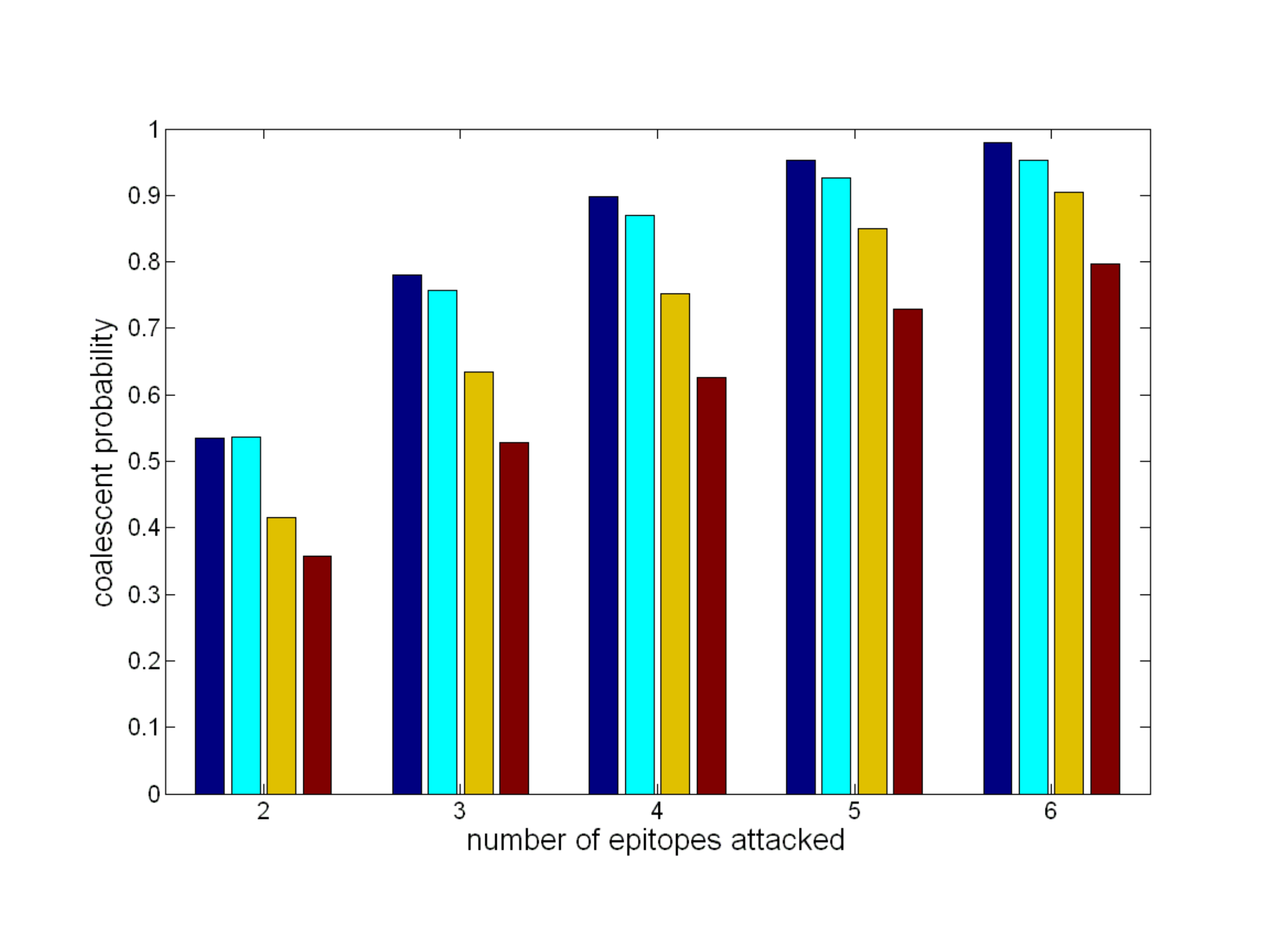}
\caption{Comparison of the probability of coalescence of two lineages for 
a full escape graph and a linear escape graph with $\gamma = 3$, $g=.1$, and
$\delk=.1$. The bars, from left to right, give the coalescent probability  under
a linear escape graph and SPR, full escape graph and SPR, linear escape graph
and LPR, and full escape graph and LPR (see Table
\ref{T:parameter} for the definition of these parameter regimes).}
\label{F:compare_coal_prob}
\end{center}
\end{figure}

\section{Linear Escape Graph}  \label{S:linear_escape_graph}
\setcounter{equation}{0}
\setcounter{claim}{0}

In this section we consider \systems for the linear escape graph under the \spl. 
Our main aim is to explain and demonstrate Theorem \ref{T:linear_theorem}.  
For notational simplicity, we set $v_i =
\underbrace{11\dots1}_{i}00\dots0$.   In this subsection we  write $e_i$ for
$e_{v_i}$ in \system.   For each variant class $\alle_i$  we define $T_i$ for
$i=1,2,\dots,\nume$ as the time at which variant $v_i$ reaches scaled population
size $\delta$,
\begin{equation} \label{E:linear_spawning_time}
T_i = \inf \{t : e_{v_i} \ge \delta\},
\end{equation}
where 
\begin{equation}  \label{E:delta_def}
\delta = \left(\frac{1}{|\log(\mu^2 \E)|}\right)^2
\end{equation} 
The value of $\delta$  can fall within a range of values, the formula above is a
specific choice within this range.  Intuitively, $\delta$ represents a
microscopic-macroscopic cutoff.   Different variants 'interact' in \systems through the $h$ equation.  
When a variant has population less than $\delta$, its impact on $h$ dynamics and
in turn on the dynamics of other variants is small and can be ignored in the
\spl.   From this perspective, the smaller $\delta$
the better.   On the other hand, a $\delta$ that is too small will make the
interval $[T_i, T_{i+1}]$ too short in the sense that the $v_{i+1} \to v_{i+2}$
mutations that drive the $i+1$th spawning period will not have finished by
$T_{i+1}$.   From this perspective, the larger $\delta$ the better.   Our choice
for $\delta$ is a middle ground between these two extremes.

In the \spl, $\delta \to 0$.  Variants with scaled population size less than
$\delta$ collapse as a percentage of the population in the \spl.  
 This is why we think of $\delta$ as a
microscopic-macroscopic cutoff.   However, if a variant has scaled
population size $\delta$ then the number of such variants, unscaled, is $\delta
\E$ which goes to $\infty$ in the $\spl$.    So while 'microscopic' variants are
few as a percentage of the population,  they may have large population sizes in
an absolute sense.

We also set
\begin{equation}
T_0 = \inf\{t : e_1(t) = \mu^2 \E \delta^2\}.
\end{equation}
$T_0$ is a special case because we set $e_1(0) = \mu \E$.

Using the $T_i$ we decompose $[0,T_\sample]$ into intervals $[T_{i-1}, T_i]$
along with initial and final intervals $[0,T_0]$ and $[T_\nume, T_\sample]$
respectively.   That this composition is valid with probability $1$
in the \spl, i.e. 
\begin{equation}  \label{E:valid}
P(T_0 < T_1 < \dots < T_\nume < T_\sample) \to 1,
\end{equation} 
will be a consequence of our analysis below.

We consider the interval $[T_i, T_{i+1}]$
for $i=0,1,\dots,\nume-2$.  The intervals $[0,T_0], [T_{\nume-1},T_\nume]$,
[$T_\nume,T_\sample]$ are handled separately.  We show below that during $[T_i,
T_{i+1}]$, only the variants $v_{i-1}, v_i, v_{i+1}$  and $v_{i+2}$ play a
significant role in the dynamics.   Table \ref{T:pop_sizes} shows the scaled population sizes
of different variants  at $T_i$ and $T_{i+1}$.   The arguments that justify
Table \ref{T:pop_sizes}  are given below, for now we focus on intuition.

\begin{table}
\begin{center} 
\begin{tabular}{|c||c|c|}  
\hline 
variant  & $e_\cdot(T_i)$ & $e_\cdot(T_{i+1})$ \\
\hline
$v_{i-1}$ & $(1-h(T_i))/h(T_i) + \text{O}(\delta)$ & O$(\delta)$\\
\hline
$v_i$ & $\delta$ & $(1-h(T_{i+1}))/h(T_{i+1}) + \text{O}(\delta)$\\
\hline
$v_{i+1}$ & O$(\mu^2 E \delta^2)$ & $\delta$\\
\hline
$v_{i+2}$ & $0$ & O$(\mu^2 E \delta^2)$\\
\hline
$v_j$ for $j >i+2$ & $0$ & $0$\\
\hline
$v_j$ for $j < i-1$ & o$(\delta)$ & o$(\delta)$\\
\hline
\end{tabular}
\end{center}
\caption{Dynamics of \systems during $[T_i, T_{i+1}]$ for a linear escape
graph.}
\label{T:pop_sizes}
\end{table}

To explain Table \ref{T:pop_sizes}, we first consider the $v_{i-1}$ and $v_i$
variants.   If only one variant type exists in whole population, say $v$, 
then \systems is composed solely of the equations for $h$ and $e_v$ and in
equilibrium we have $e_v \approx (1-h)/h$. Examining Table \ref{T:pop_sizes}, we see that at $T_i$,
$v_{i-1}$ is roughly at this equilibrium,  meaning that it is the dominant
variant in the HIV population.    On the other hand, $v_i$ variants at time
$T_i$ are few since $\delta \ll (1-h)/h$.    However, by time $T_{i+1}$, the situation
has flipped with $v_i$ dominating the  population and $v_{i-1}$ pushed to low
levels. Intuitively, the $v_i$ variants  are more fit and push out the $v_{i-1}$
variants during $[T_i, T_{i+1}]$.

Now consider the $v_{i+1}, v_{i+2}$ variants.   Recalling that $\E e_{i+1}(T_i)$
gives the number, unscaled, of $v_{i+1}$ variants. We first note that 
\begin{equation}
\E e_{i+1}(T_i) = O(\mu^2 \E^2  \delta^2) \ll \frac{1}{\mu}.
\end{equation}
The $\ll$ directly above is justified in the \spls since $\mu^3 \E^2 \to 0$. 
Since the probability of mutation is $\mu$, the inequality above shows that at
$T_i$ the rate of $v_{i+1} \to v_{i+2}$ mutations goes to $0$ in the $\spl$. 
However, notice that $\E e_{i+1}(T_i) \approx O(\mu^2 \E^2) \to \infty$, meaning
that $e_{i+1}$ dynamics are deterministic at time $T_i$. 

Turning to $v_{i+2}$ we see that at $T_i$ no such variants exist.   However, by
time $T_{i+1}$, there are enough such variants to make their dynamics
deterministic.  Connecting to our comments in the Results section, $[T_i,
T_{i+1}]$ is the $i+1$th spawning phase or, slightly more explicitly, the
$v_{i+1} \to v_{i+2}$ spawning phase.

Finally we note that Table \ref{T:pop_sizes} shows that all other variant types
are of negligible population size.   $v_j$ with $j > i+2$ have yet to arise and
$v_j$ with $j<i-1$ have been previously driven to low levels by fitter variants.

Table \ref{T:pop_sizes} provides the outlines of an iteration, as we proceed
through different values of $i$, that allows us to analyze the stochastic
dynamics of \system.  The key to deriving the table is an 
estimate of
$e_{i+2}(T_{i+1})$, the pop value of $v_{i+2}$. In subsection \ref{S:dynamics}
we show
\begin{equation}  \label{E:key}
\frac{e_{i+2}(T_{i+1})}{\mu \E^2 \delta^2} \to P_{i+2},
\end{equation}
where
\begin{equation}  \label{E:P_v}
P_{i+2} = \left(\frac{\delk}{k_{i}}\right)^2 \S(\frac{1}{2}, 1,
\pi^2\left(\frac{\delk}{k_i}\right)),
\end{equation}
and where $\S(\alpha, \beta, c)$ is the stable distribution with
index $\alpha$, skewness parameter $\beta$, and scale factor $c$ 
\cite{Nolan_book}. 

In subsection \ref{S:dynamics}, we provide the arguments that justify Table
\ref{T:pop_sizes} and (\ref{E:key}).   Then in section
\ref{S:lineage_construction}, we use the results of section
\ref{S:dynamics}, which center on the dynamics of \system, to  demonstrate our
lineage result, Theorem \ref{T:linear_theorem}.

\subsection{Dynamics}  \label{S:dynamics}

The goal of this section is to prove Proposition \ref{P:linear_dynamics} which
is a precise version of Table \ref{T:pop_sizes} and (\ref{E:key}).  

\begin{proposition} \label{P:linear_dynamics}
Consider \systems on a linear escape graph.  Assume that at $T_i$ for
$i=1,\dots,\nume-2$
\begin{enumerate}
  \item $P(e_j(T_{i}) = 0) \to 1$ for $j \ge i+2$.
  \item $\frac{e_{i+1}(T_{i})}{\mu \E^2 \delta^2} \to P_{i+1}$
  \item $e_{i}(T_{i}) = \delta$
  \item $e_{i-1}(T_{i}) = \frac{1 - h(T_{i})}{h(T_{i})} + O(\delta)$.
  \item $e_{j}(T_{i}) = o(\delta)$ for $j < i-1$.
\end{enumerate}
Then at time $T_{i+1}$ we have the following conclusions
\begin{enumerate}
  \item $P(e_j(T_{i+1}) = 0) \to 1$ for $j \ge i+3$.
  \item (\ref{E:key}) holds
  \item $e_{i+1}(T_{i+1}) = \delta$
  \item $e_{i}(T_{i+1}) = \frac{1 - h(T_{i+1})}{h(T_{i+1})} + O(\delta)$.
  \item $e_{j}(T_{i+1}) = o(\delta)$ for $j < i$.
\end{enumerate}
\end{proposition}
Conclusion $3$ of Proposition \ref{P:linear_dynamics} holds by the definition of
$T_{i+1}$.  Conclusions $4$ and $5$ could be phrased in terms of the \spl, for
instance  $e_j(T_i)/\delta \to 0$, but the o() notation is, to our taste,
clearer.  

We can apply Proposition \ref{P:linear_dynamics} recursively to
characterize the dynamics at each time $T_i$ for $i=1,\dots,\nume-2$.  
The case $i=0$, which we must consider to start the recursion, is handled
through the same arguments that give Proposition \ref{P:linear_dynamics}, except
that our assumptions are slightly different.   Namely, at $T_0$ we have
$e_1(T_0) = \mu^2 \E \delta^2$ by definition of $T_0$, which parallels
assumption $2$ of Proposition \ref{P:linear_dynamics} and $e_0(T_0) = (1-h)/h
+ o(\delta)$ which parallels assumptions $3$.  Assumption $1$ of Proposition
\ref{P:linear_dynamics} holds for the $i=0$ case, while assumptions $4$ and $5$
are not applicable.   

To explain Proposition \ref{P:linear_dynamics}, we split
$[T_i, T_{i+1}]$ into two time intervals:  $[T_i, T_i^h]$ and $[T_i^h,
T_{i+1}]$. In Lemma \ref{L:transition_phase}, we show that during $[T_i, T_i^h]$ the
$v_{i}$ variant displaces the $v_{i-1}$ variant as the dominant variant, as 
alluded to in Table \ref{T:pop_sizes} and the accompanying discussion.  This
transition happens quickly, so that $T_i^h - T_i$ is small.   As a result, the
$v_{i+1}$ population does not grow in size much and no $v_{i+1} \to v_{i+2}$
mutations occur.  Through the arguments of Lemma \ref{L:spawning_phase}, we show
that $v_{i+1}, v_{i+2}$ dynamics on $[T_i^h, T_{i+1}]$ obey the generalized LD dynamics
discussed in the Results section.  
 
\begin{lemma}  \label{L:transition_phase}
Adopt the same assumptions stated in Proposition \ref{P:linear_dynamics}. 
Set $T_i^h = T_i + (3/\delk + \frac{2}{\eigconstant})|\log(\delta)|$ where
$\eigconstant$ is given in (\ref{E:eigconstant}). Then,
\begin{enumerate}
  \item for $t \in [T_i, T_i^h]$, $\mu \E e_{i+1}(t) \to 0$,
  \item $P(e_{i+2}(T_i^h) = 0) \to 1$,
  \item $e_{i}(T_i^h) = \frac{1 - h(T_i^h)}{h(T_i^h)} + O(\delta)$
  \item $e_{i-1}(T_i^h) < \delta$.
\end{enumerate}
\end{lemma}

Conclusions $1$ and $2$ of Lemma \ref{L:transition_phase} guarantee that no
$v_{i+1} \to v_{i+2}$ mutations occur during $[T_i,T_i^h]$ and follow from the
small population size of $v_{i+1}$ variants at $T_i$.  Indeed,  by assumption
we have $e_{i+1}(T_i) = O(\mu^2 \E \delta^2)$.    The number of $v_{i+1}$
variants grow exponentially on $[T_i, T_i^h]$. However, $T_i^h - T_i = O(|\log(\delta)|)$  and so despite their exponential growth the number of $v_{i+1}$ variants remains small.    More precisely,
referring to \system, we note that $h$ is bounded above by $1$ and so we can
bound the exponential growth rate of $e_{i+1}$ by $\gamma$,
\begin{equation}
\frac{de_{i+1}}{dt} \le \gamma e_{i+1}.
\end{equation}
Integrating the above equation and using the assumption on
$e_{i+1}(T_i)$ gives,
\begin{align}  \label{E:e_bound}
e_{i+1}(t) &  \le O(\mu^2 \E \delta^{2 + \gamma}) \to 0,
\end{align}
where we have used the observation 
\begin{equation}
\mu^2 E \ll (\mu^2 E) \mu E = \mu^3 \E^2 \to 0
\end{equation} 
to justify convergence to $0$ in the \spl.  This gives conclusion $1$.

Conclusion $2$ follows almost directly from conclusion $1$.  We recall from
(\ref{E:final_system}) that $v_{i+1} \to v_{i+2}$ mutations arise at rate
$\mu \E \gamma h e_{i+1}$.  The number of such mutations in the interval
$[T_i, T_i^h]$ is then a Poisson process with mean,
\begin{equation}
O(\mu \E)
\int_{T_i}^{T_i^h} ds h(s) e_{i+1}(s) 
\end{equation}
Plugging in our bound from (\ref{E:e_bound}) shows the mean number of mutations
to be bounded by $O(\mu^3 E^2)$, here we've ignored $\delta$ factors.  Taking
the \spls gives conclusion $2$.

To explain conclusions $3$ and $4$ of Lemma \ref{L:transition_phase} we notice
that only variants $v_i$ and $v_{i-1}$ are of order greater than 
$\delta$ throughout $[T_i, T_i^h]$.   Ignoring the other variants
then, our ODE (\ref{E:final_system}) reduces to three equations involving $e_i, e_{i-1},
h$.  Since variant $v_i$ has one less epitope exposed to CTL attack than
$v_{i-1}$, it will eventually push the $v_{i-1}$ to extinction.   Initially, 
$e_i(T_i) = \delta$.  Since the CTL kill rate of $v_{i-1}$ variants is
$\delk$ greater than those of $v_i$ variants, initially the $v_i$ variants grow
exponentially with rate $\delk + O(\delta)$.  It then takes $O(\frac{1}{\delk}
|\log(\delta)|)$ time for the $v_i$ population to rise to $O(1)$ levels, push
out the $v_{i-1}$ population, and near equilibrium with respect to $h$.  This
explains the order of $T_i^h$.  The exact form of $T_i^h$ is explained in
section \ref{AS:transition} of the appendix as are the technical details
demonstrating conclusions $3$ and $4$.

Now we consider $[T_i^h, T_{i+1}]$ through the following lemma.

\begin{lemma}  \label{L:spawning_phase}
Assume the conclusions of Lemma \ref{L:transition_phase}.  Then,
\begin{enumerate}
  \item $e_{i+1}(T_{i+1}) = \delta$
  \item $\frac{e_{i+2}(T_{i+1})}{\mu^2 \E \delta^2} \to P_{i+2}$,
  \item $e_{i}(T_{i+1}) = \frac{1 - h(T_{i+1})}{h(T_{i+1})} + O(\delta)$.
  \item $e_{j}(T_{i+1}) = o(\delta)$ for $j < i$.
\end{enumerate}
\end{lemma}

Conclusion $1$ of Lemma \ref{L:spawning_phase} follows from the definition of
$T_{i+1}$.   Assuming conclusion $1$, we see that variants $v_{i+1}$,
$v_{i+2}$ remain at O$(\delta)$ levels throughout $[T_i^h, T_{i+1}]$.   As a
result, the approximate equilibrium of $e_i, h$ which exists at $T_i^h$ is
maintained. Further, variants $v_j$ for $j < i$ continue to drop in number as they are less fit than $v_i$ variants.  These
observations justify Conclusions $3$ and $4$.

We have left to consider Conclusion $2$.   From (\ref{E:final_system}) we
have the following ODEs for $e_{i+1}, e_{i+2}$,
\begin{gather} \label{E:spawning_ODE}
de_{i+1} = \gamma e_{i+1} (h - \frac{k_{i+1}}{\gamma}), \\ \notag
de_{i+2} = \frac{1}{\E} \bigg( dP(\gamma h \E e_{i+2} h) 
	- dP(k_{i+2} \E e_{i+2}) 
	+  dP(\mu \gamma h \E e_{i+1}) \bigg).
\end{gather}
Note that $e_{i+1}$ is given by a deterministic ODE because $\E e_{i+1}(T_i)
\to \infty$.   By conclusion $3$, since $e_i,h$ are near equilibrium, we know $h
= k_i/\gamma + O(\delta)$. Plugging this result into (\ref{E:spawning_ODE}) gives,
\begin{gather} \label{E:spawning_ODE_2}
de_{i+1} = (\delk + O(\delta)) e_{i+1} , \\ \notag
de_{i+2} = \frac{1}{\E} \bigg( dP((k_i + O(\delta)) \E e_{i+2})
     - dP(k_{i+2} \E e_{i+2}) 
	+  dP(\mu (k_i + O(\delta)) \E e_{i+1}) \bigg).
\end{gather}
If we label the number of $v_{i+2}$ variants as $e^\#_{i+2}$,
by our scaling $e^\#_{i+2} = \E e_{i+2}$, the $e_{i+2}$ equation
in (\ref{E:spawning_ODE_2}) transforms into,
\begin{equation}
de^\#_{i+2} =  dP((k_i + O(\delta)) e^\#_{i+2})
     - dP(k_{i+2} e^\#_{i+2}) 
	+  dP(\mu (k_i + O(\delta)) \E e_{i+1}),
\end{equation}
and we find
that $v_{i+2}$ variants evolve according to a binary branching process with
birth rate $k_i + O(\delta)$, death rate $k_{i+2}$, and mutation rate that creates new
$v_{i+2}$ variants $\mu (k_i + O(\delta)) \E e_{i+1}$. 

Ignoring
the $O(\delta)$ term in the rates, the growth rate for $v_{i+2}$ variants,
which we label as $r_{i+2}$, is $r_{i+2} = k_i
- k_{i+2} = 2\delk$. Considering
$v_{i+1} \to v_{i+2}$ mutations, we have for the rate
\begin{align} \label{E:immigration_rate}
\text{rate $v_{i+1} \to v_{i+2}$ at time $t$} & \approx 
	\mu k_i \E e_{i+1}(t) 
\\ \notag
	& = \mu k_i \E e_{i+1}(T_i^h) \exp[\delk(t - T_i^h)].
\end{align}   

By the assumptions of Lemma \ref{L:spawning_phase}, there are no $v_{i+2}$
variants at time $T_i^h$.  The $v_{i+2}$ population arises from mutations in the
$v_{i+1}$ population which expands at rate $\delk$.  Such mutations
produce $v_{i+2}$ cells that then expand at rate $2 \delk$, precisely the
generalized LD dynamics mentioned in the Results section.

We define $\LD_\classic(t)$ to be the number of mutants at time $t$ for the LD
model in which mutants and wild types grow at the same rate. In
\cite{Kepler_2001_TPB, Mohle_2005_J_App_Prob} the following  asymptotic formula
was derived under the further assumptions that wild types grow deterministically
and mutants grow stochastically, but with no death events.  
\begin{equation}  \label{E:Mohle_LD}
\LD_\classic(t) \approx m \log(m) + m \S(1,1,\frac{\pi}{2})
\end{equation}
where $m$ is  the expected number of mutations on the time interval $[0,t]$  for a wild
type population that is of size $1$ at time $0$.   The relative error of
(\ref{E:Mohle_LD}) goes to $0$ as $m \to \infty$.

In contrast to $\LD_\classic(t)$, we let $\LD_{2}(t)$ be the number of
mutants at time $t$ for the generalized LD model in which mutants grow at
double the rate of wild types.  As in the case of $\LD_\classic$, we will assume
that wild types grow deterministically, matching the deterministic growth of
$v_{i+1}$ variants as they spawn.   However, for $\LD_2$ we will assume that
mutants have non-zero birth and death rates, corresponding to the situation for
$v_{i+2}$ variants.   In section \ref{AS:spawning} of the appendix, using
generalizations of the techniques found in \cite{Mohle_2005_J_App_Prob}, we show
\begin{equation}  \label{E:Mohle_LD_2}
\LD_2(t) \approx m^2 
	S(\frac{1}{2}, 1,\pi^2\left(\frac{\delk}{k_i}\right)).
\end{equation}
As for (\ref{E:Mohle_LD}), the relative error of (\ref{E:Mohle_LD_2}) goes to
$0$ as $m \to \infty$.

$e^\#_{i+2}(t)$ is approximately an $\LD_2(t)$ process on $[T_i^h, T_{i+1}]$ and
becomes exactly so in the \spl.   In this setting, $m$ is the expected number of $v_{i+1}
\to v_{i+2}$ mutations during $[T_i^h, T_{i+1}]$.   We show in section
\ref{AS:spawning} of the appendix, $m \approx (\frac{k_i \delta}{\delk}) \mu
\E$. Plugging this value of $m$ into (\ref{E:Mohle_LD_2}), we find
\begin{equation}  \label{E:linear_pop_value_asymp}
e^\#_{i+2}(T_{i+1}) = \left(\frac{k_i}{\delk} \right)^2 \mu^2 \E^2 \delta^2
\S(\frac{1}{2}, 1, \pi^2\left(\frac{\delk}{k_i}\right)).
\end{equation} 
Recalling that $e_{i+2}(t) = e^\#_{i+2}(t)/\E$ gives conclusion $2$ of Lemma
\ref{L:spawning_phase}.

\subsection{Lineage Construction}  \label{S:lineage_construction}

To demonstrate Theorem \ref{T:linear_theorem}, we need some additional lineage
notation.  Recall that we consider $n$ lineages, $\lin_j$ for $j=1,2,\dots,n$,
corresponding to the $n$ sampled cells.   We let $\lin_j(t)$ be the ancestral
cell at time $t$ of sample cell $j$.  (To make this precise we could number the cells
in our process as they are born, and then $\lin_j(t)$  would map to
$\mathbb{N}$, but we will not make this explicit.)  We
let $\Var(\lin_j(t))$ be the variant type of $\lin_j(t)$.    For example,
$\Var(\lin(T_\sample)) = \textbf{11\dots1}$.  We write $\lin_j$, dropping the
time dependence, when we are considering the lineage over a range of times.

To combine the separate lineages into a genealogy, we need to identity
mutation and coalescent events. A \textit{mutation event} on $\lin_j$ occurs at
time $t$ if the variant of $\lin_j$ changes at time $t$, more precisely
$\Var(\lin_j(t-)) \ne \Var(\lin_j(t))$.  Given two lineages $\lin_j, \lin_k$,
the lineages have \textit{coalesced by time $t$} if $\lin_j(t) = \lin_k(t)$
and we say that the lineages \textit{coalesced at time $t$}  if for $t' > t$,
$\lin_j(t') \ne \lin_k(t')$.

We prove Theorem \ref{T:linear_theorem} by considering mutation and
coalescent events on the interval $[T_i, T_{i+1}]$. 
Lemma \ref{L:lineage_prior_bottle} sets up this analysis by showing that no
mutation or coalescent events occur on $[T_{\nume-1}, T_\sample]$.  This allows
us to consider the interval $[T_{\nume-2}, T_{\nume-1}]$ with all lineages of
type $v_\nume$ at $T_{\nume-1}$.

Lemma \ref{L:lineage_bottle} provides two results for the general setting of an 
interval $[T_i, T_{i+1}]$, assuming that all lineages are of type $v_{i+2}$ at
$T_{i+1}$.  First, by time $T_i$, all lineages are of type $v_{i+1}$.  This
result implies that each lineage must experience a $v_{i+1}
\to v_{i+2}$ mutation during $[T_i, T_{i+1}]$.  Second, two lineages, $\lin_j,
\lin_k$, coalesce during $[T_i, T_{i+1}]$ if and only  if their associated cells
at $T_{i+1}$, $\lin_j(T_{i+1})$, $\lin_k(T_{i+1})$, descend from  the same $v_{i+1} \to
v_{i+2}$ mutation.  In other words, the lineages  coalesce if there is a cell
that is of variant type $v_{i+1}$ which  produces a child cell of type $v_{i+2}$
from which $\lin_j(T_{i+1})$, $\lin_k(T_{i+1})$  are both descended.  Lemma
\ref{L:lineage_bottle} reduces the analysis of  coalescent events on $[T_i,
T_{i+1}]$ to the analysis of mutation events and  the number of their
descendants at $T_{i+1}$.  
 
(\ref{E:linear_pop_value_asymp}) gives an asymptotic
description for the number of descendants at $T_{i+1}$ produced by
all $v_{i+1} \to v_{i+2}$ mutations during $[T_i, T_{i+1}]$.  In other words, the
pop value of $v_{i+2}$.   Lemma \ref{L:coal_prob} describes $\Pi(T_i)$ given
$\Pi(T_{i+1})$ by considering each such mutation separately.   To do this, we
decompose (\ref{E:linear_pop_value_asymp}) into a collection of single mutation
results as described in (\ref{E:pop_decomp_linear}).  Through this
decomposition, by exploiting Lemma \ref{L:lineage_bottle}, we characterize coalescent events on $[T_i, T_{i+1}]$. 
By repeatedly applying Lemmas \ref{L:lineage_bottle} and \ref{L:coal_prob} we
can characterize coalescent events on $[0, T_\sample]$.

\begin{lemma} \label{L:lineage_prior_bottle}
For $t > T_{\nume-1}$ and $j,k = 1,2,\dots,n$, $\Var(\lin_j(t)) = v_{\nume}$
 (no mutation events occur) and $\lin_j(t) \ne \lin_k(t)$ if $j \ne k$ (no
coalescent events occur).
\end{lemma}

\begin{proof}
Consider the probability that the
$\lin_j$ lineage experiences a mutation at time $t > T_{\nume-1}$.    
For such an event to occur, a $v_{\nume-1} \to v_\nume$ mutation must  occur and 
the resultant  $v_\nume$ variant must be in the $\lin_j$ lineage.   
The rate of $v_{\nume-1} \to v_\nume$ mutations is given by $\mu \gamma \E h
e_{\nume-1}(t)$ which is trivially bounded by $O(\mu \E)$.  
By symmetry the $v_{\nume}$ variant resulting from a mutation is in
$\lin_j$ with probabiliy $O(\frac{1}{\E e_\nume(t)})$.  By Proposition
\ref{P:linear_dynamics}, for $t > T_{\nume-1}$ this probability is bounded above
by $O(\frac{1}{\mu^2 \E^2 \delta^2})$.   From this we have,
\begin{align}
P(\text{no mutation event on } & [T_{\nume-1},T_\sample])
\\ \notag
	& = \exp[-\int_{T_{\nume-1}}^{T_\sample} ds O(\frac{1}{\mu \E \delta^2})]
\\ \notag
	& = \exp[-O(\frac{1}{\mu \E \delta})] \to 1.
\end{align}
 In the last line above we have used the result $T_\sample - T_{\nume-1} =
O(\delta)$.   To see this note that after $T_{\nume-1}$, the $v_\nume$
 variants expand deterministically.   Arguments similar to those used in
 Proposition \ref{P:linear_dynamics} show that $v_\nume$ will push
 $v_{\nume-1}$ to $O(\delta)$ levels in $O(\delta)$ time.  
 
 The argument for no
  coalescent events is similar. For $\lin_j,\lin_k$ to coalesce at time $t$, a
 $v_\nume$  variant must give birth to a new $v_\nume$ child cell, which occurs
 with rate  $O(\E e_\nume(t))$ and $\lin_j(t), \lin_k(t)$ must be, in no
 particular  order, precisely these parent and child cells, which occurs with
 probability $O(1/(\E e_\nume(t))^2)$.  This leads to,
\begin{align}
P(\text{no coalescent event on }& [T_{\nume-2}, T_\sample])
\\ \notag
 & = \exp[-\int_{T_{\nume-1}}^{T_\sample} ds O(\frac{1}{\E e_\nume(s)})]
\\ \notag
	& = \exp[-O(\frac{\delta}{\mu^2 \E^2 \delta^2})] \to 1.
\end{align}
\end{proof}

As mentioned, Lemma \ref{L:lineage_prior_bottle} allows us to consider $[T_i,
T_{i+1}]$ under the assumption that all lineages are of type $v_{i+2}$ at
$T_{i+1}$.  With this in mind, we 
introduce the following definitions.
\begin{gather}
\tau_j = \inf\{t : \Var(\lin_j(t)) = v_{i+2}\}, \\ \notag
a_j = \lin_j(\tau_j).
\end{gather}
$\tau_j$ is the time of the $v_{i+1} \to v_{i+2}$ mutation event on
$\lin_j$ and $a_j$ is the specific infected cell, a $v_{i+1}$ variant, that
produces the cell $\lin_j(\tau_j)$, a $v_{i+2}$ variant.  We use $a_j$ as a
mnemonic for 'ancestor' since $a_j$ will be the ancestor of $\lin_j(T_{i+1})$.

\begin{lemma} \label{L:lineage_bottle}
Let $j,k \in \{1,2,\dots,n\}$ and assume $\Var(\lin_j(T_{i+1})) = v_{i+2}$ for
all $j$.   Then $\tau_j \in [T_{i}, T_{i+1}]$ (a $v_{i+1} \to v_{i+2}$ mutation
occurs on $[T_i, T_{i+1}]$) and $\Var(\lin_j(T_i)) = v_{i+1}$. Further for $j
\ne k$, $\lin_j(T_{i}) = \lin_k(T_{i})$  (a coalescent event has occurred) if and
only if $a_j = a_k$.
\end{lemma}

\begin{proof}
Proposition \ref{P:linear_dynamics} shows that no $v_{i+2}$ variants exist at
time $T_{i}$, so $\Var(\lin_j(T_i)) \ne v_{i+2}$.   This immediately implies
that $\tau_j \in [T_{i},T_{i+1}]$.   For $t \in [T_i, \tau_j]$,
 essentially the same arguments that gave Lemma
\ref{L:lineage_prior_bottle} show that no $v_{i+1}$ variant lineages experience
mutation events prior to $T_i$.   Consequently, we can conclude
$\Var(\lin_j(T_i)) = v_{i+1}$.
 
Now we consider coalescent events.  If $\tau_j = \tau_k$ then  we have $a_j =
a_k$ as required by the lemma.  So now assume $\tau_j > \tau_k$,  we want to
show that the two lineages do not coalesce prior to $T_i$.  Since $\tau_j \ne
\tau_k$, $\lin_j$ and $\lin_k$ cannot  
coalesce during $[\tau_j, T_{i+1}]$, 
otherwise we would necessarily have $\tau_j = \tau_k$.  Further since $\lin_j$ and $\lin_k$ are of  different variant type
during  $(\tau_k, \tau_j]$,  no coalescent event occurs  on $(\tau_k, \tau_j]$.  
On the interval $[T_i, \tau_k)$, $\lin_k,\lin_j$ are of type $v_{i+1}$  and the
same arguments that gave Lemma \ref{L:lineage_prior_bottle}  show that $v_{i+1}$
variants do not coalesce prior to $T_i$.   We are left with the possibility of a
coalescent event at time $\tau_k$.  This would mean that $\lin_j(\tau_k-)$,
which is of type $v_{i+1}$, produces a mutant child cell that is of type
$v_{i+2}$ which is precisely $\lin_k(\tau_k)$.    However, the mutation event
associated with $\lin_k(\tau_k)$  is equally likely to be produced by any variant $v_{i+1}$ at
time $\tau_k$.      The probability that $\lin_j(\tau_k-)$ is the cell chosen
is $\frac{1}{E e_{i+1}(\tau_k)}$ which is bounded as follows,
\begin{equation}
\frac{1}{\E e_{i+1}(\tau_k)} < \frac{1}{\E e_{i+1}(T_{i})} = \frac{1}{O(\mu^2
\E^2 \delta^2)} \to 0.
\end{equation}
\end{proof}

The random partition $\Xi_{A,i}$ characterizes the coalescent events that occur
on $[T_i, T_{i+1}]$.   From Lemma \ref{L:lineage_bottle}, we know that
coalescent events are associated with $v_{i+1} \to v_{i+2}$ mutations during $[T_i,
T_{i+1}]$.   Mutations that occur relatively early in this time
interval will, on average, produce more descendants at $T_{i+1}$.  Consequently
a lineage, $\lin_j$, is more likely to descend from a mutation that occurs
relatively early.   The parameter $A$ considers mutations that happen on the
interval $[T_i, T_A]$ where $T_A$ is  defined as the time at which the mean
number of $v_{i+1} \to v_{i+2}$ mutations that are expected to occur is
precisely $A$.  Lemma \ref{L:coal_prob} shows that the error in the
approximation $\Xi_{A,i}$ collapses as $A \to \infty$.

\begin{lemma} \label{L:coal_prob}
Let $\Delta$ be a fixed partition of $\Pi(T_{i+1})$ and $A$ a fixed constant.  
Then,
\begin{equation}
\lim P(\Pi(T_i) = \Delta) = P(\Xi_{A,i}(\Pi(T_{i+1})) = \Delta) +
O(\frac{1}{A}).
\end{equation}
\end{lemma}
 
\begin{proof}
Let $m$ be the mean number of $v_{i+1} \to v_{i+2}$ mutations on $[T_i,
T_{i+1}]$ and $K_\all$ a Poisson r.v. with mean $m$.  Then there are $K_\all$
such mutations.  Numbering the mutations in some arbitrary way, we let $
e_{i+2}^{\#,(q)}(T_{i+1})$ be the number, unscaled, of ancestors at $T_{i+1}$
that descend from mutation $q$.    The total $v_{i+2}$ population at $T_{i+1}$
is then given by,
\begin{equation}  \label{E:m_simple}
\E e_{i+2}(T_{i+1}) = \sum_{q=1}^{K_\all} e_{i+2}^{\#,(q)}(T_{i+1})
\end{equation} 

We split the interval $[T_{i}, T_{i+1}]$ into the intervals $[T_{i}, T_A]$ and
$[T_A, T_{i+1}]$.   We let $m_\error$ be the mean number of mutations in
$[T_A,T_{i+1}]$. As mentioned, $T_A$ is chosen so that the mean number of
mutations in $[T_{i}, T_A]$ is $A$.  Then, we have $m = A + m_\error$.  Let
$\M_A$ and $\M_\error$ be the set of mutations that occur on $[T_{i}, T_A]$, $[T_A,
T_{i+1}]$ respectively. Then (\ref{E:m_simple}) becomes 
\begin{equation} \label{E:m_two}
\E e_{i+2}(T_{i+1}) = \sum_{q \in \M_A} e_{i+2}^{\#, (q)}(T_{i+1}) 
	+ \sum_{q' \in \M_\error} e_{i+2}^{\#, (q')}(T_{i+1}),
\end{equation}
 
In section \ref{AS:coal_prob} of the appendix we show,
\begin{equation}    \label{E:mut_ratio}
\frac{\sum_{q \in \M_\error} e_{i+1}^{\#, (q)}(T_{i+1})}{\sum_{q \in \M_A} 
e_{i+2}^{\#,(q')}(T_{i+1})} \to O(\frac{1}{A}),
\end{equation}
and (\ref{E:m_two}) reduces to
\begin{equation} \label{E:m_reduced}
\E e_{i+2}(T_{i+1}) \approx \sum_{q \in \M_A} e_{i+2}^{\#,(q)}(T_{i+1}).
\end{equation}
A simple computation
shows that the number of $\M_\error$ mutations is $O(\mu E)$  while the number
of $\M_A$ mutations is almost by definition $O(A)$.    However,  $\M_A$
mutations happen early allowing their descendant population to expand at rate
$2\delk$ for a longer time than $M_\error$ mutations.  On the other hand, the
lateness of $\M_\error$ mutations means there will be more such mutations
because the $v_{i+1}$ population expands at rate $\delk$.  Since the $v_{i+1}$
population expands at half the rate of the $v_{i+2}$ population, the descendants
of early mutations dominate.

From (\ref{E:m_reduced}), in the \spls each lineage cell $\lin_j(T_{i+1})$
must descend from one of the $\M_A$ mutations.  Notice that $K$ in
Definition \ref{D:linear_xi} is precisely the number of $\M_A$ mutations.   By
Lemma \ref{L:lineage_bottle} $\lin_j, \lin_k$ coalesce  during $[T_i, T_{i+1}]$
if and only if they descend from the same mutation.   Lineage $j$ descends from
mutation $q$ with probability,
\begin{equation}  \label{E:e_ratio}
\frac{e_{i+2}^{\#,(q)}(T_{i+1})}{ \sum_{q'=1}^{K} 
e_{i+2}^{\#,(q')}(T_{i+1})}.
\end{equation}
In section \ref{AS:coal_prob} of the appendix through branching process
asymptotics we show
\begin{equation}   \label{E:proportional_gamma}
e_{i+2}^{\#,(q)}(T_{i+1}) \to C \Gamma^{(i)}
\end{equation}
where $C$ is a constant that is independent of $q$, and 
$\Gamma^{(i)}$ is as defined in the Results section just prior to Definition
\ref{D:linear_xi}.   Combining (\ref{E:m_reduced}) and
(\ref{E:proportional_gamma}) gives the pop value decomposition
(\ref{E:pop_decomp_linear}) stated in the Results section.  The constant $C$
cancels out in the ratio (\ref{E:e_ratio}) and the resultant formula is precisely the 'coloring' probability given in Definition \ref{D:linear_xi}.   Putting all these observations together proves the lemma.
\end{proof}

As mentioned in the Results section, we do not sample pop values even though
(\ref{E:key}) would allow it.   Crucially, if we sampled pop values then to
construct lineages we would need to determine $e_{i+2}^{\#,(q)}(T_{i+1})$
conditioned on the pop value $e_{i+2}(T_{i+1})$.   We do not know how to sample from this
conditional distribution, so instead we sample $\Gamma^{(i)}$ which is
proportional to $e_{i+2}^\#(T_{i+1})$ (by (\ref{E:proportional_gamma})) and
thereby implicitly sample $e_{i+2}(T_{i+1})$ through the decomposition
(\ref{E:pop_decomp_linear}).  Finally we note that the $C$ in
(\ref{E:proportional_gamma}) is difficult to compute as it depends on $h$. 
Fortunately, lineage construction can proceed without knowing $C$.

Lemmas \ref{L:lineage_prior_bottle}-\ref{L:coal_prob} demonstrate Theorem
\ref{T:linear_theorem}.  Starting at $T_\sample$, Lemma
\ref{L:lineage_prior_bottle} shows that no coalescent events happen back to
$T_{\nume-1}$.  The on each $[T_i, T_{i+1}]$ for $i = 0,1,\dots,\nume-2$ the
coalescent events are described by $\Xi_{A,i}$.   Concatenating these
$\Xi_{A,i}$ takes us back to $T_0$.  By our assumption on initial conditions, a
lemma very similar to \ref{L:lineage_prior_bottle} can be proved showing that no
coalescent events occur on $[0,T_0]$.

\section{Full Escape Graph}  \label{S:full_escape_graph}
\setcounter{equation}{0}
\setcounter{claim}{0}

In this section we generalize the
arguments used in Section \ref{S:linear_escape_graph} for the linear escape
graph to the full escape graph.  As we did for the linear escape graph, we
divide the dynamics on the full escape graph  into time intervals $[T_{i},
T_{i+1}]$.   For the linear escape graph, the  $i$th variant class, $\alle_i$,
is composed of a single variant $v_i$ and the $T_i$  are defined by $e_i(T_i) =
\delta$ in (\ref{E:linear_spawning_time}).  To generalize this definition to the
full escape graph,  we let $T_i$ be the first  time any of the variant
populations in class $i$ reaches a scaled population size $\delta$:
\begin{equation} \label{E:full_spawning_time}
T_i = \inf\{t : \exists v \in \alle_i \text{ such that } e_v(t) \ge \delta\}.
\end{equation}

Sections \ref{S:full_dynamics} and \ref{S:full_lineage_construction} generalize
the results for the linear escape graph to the
full escape graph.  The arguments are similar, so we emphasize the
novel ideas that apply to the full escape graph case.  As in section
\ref{S:linear_escape_graph}, we base our analysis on consideration of an
interval $[T_i, T_{i+1}]$.

\subsection{Dynamics}  \label{S:full_dynamics}

During $[T_i, T_{i+1}]$, for the linear escape graph, only $v_{i+1}$ variants
spawn $v_{i+2}$ variants.  Recalling that $\P(v)$ is the set of variant types
that can mutate into variant $v \in \alle_{i+2}$, for the full escape graph all
$v' \in \P(v)$ spawn $v$ variants.  For example $\textbf{1100}$ can be spawned by
$\textbf{1000}$ or $\textbf{0100}$.   Consequently, the pop value
result, (\ref{E:key}), must be generalized as follows,
\begin{equation}  \label{E:full_key}
\frac{e_v(T_{i+1})}{\mu^2 \E \delta^2} \to P_v
\end{equation}
where 
\begin{equation} \label{E:full_pop}
P_v = \sum_{v' \in \P(v)} P_{v' \to v}
\end{equation}
and
\begin{equation} \label{E:full_pop_partial}
P_{v' \to v} = \left(\frac{P_{v'}}{P_{\max,i+1}}\right)^2
\left(\frac{\delk}{k_i}\right)^2 \S(\frac{1}{2}, 1,
\pi^2\left(\frac{\delk}{k_i}\right)).
\end{equation} 
where 
\begin{equation}
P_{\max,i+1} = \max_{v'' \in \alle_{i+1}} P_{v''}.
\end{equation}
With $P_v$ generalized from (\ref{E:P_v}) to (\ref{E:full_pop_partial}), a
version of Proposition \ref{P:linear_dynamics} generalized in a similar way
holds for the full escape graph.   Namely, during $[T_i, T_{i+1}]$ variants in
$\alle_{i}$ sweep to dominance,  pushing the $\alle_{i-1}$ to O($\delta$)
levels.   Concurrently, $\alle_{i+1}$ variants rise to
$O(\delta)$ levels while spawning $\alle_{i+2}$ variants.  Spawning events occur
for every $v \in \alle_{i+2}$, $v' \in \P(v)$ pair and conclusion $2$ of
Proposition \ref{P:linear_dynamics} is generalized to included all such
$P_{v' \to v}$ pop values as described in
(\ref{E:full_key})-(\ref{E:full_pop_partial})

Notice that the difference between (\ref{E:full_pop_partial}) and (\ref{E:P_v})
is the factor $(P_{v'}/P_{\max,i+1})^2$ in (\ref{E:full_pop_partial}).   As this
difference is the main technical novelty in moving from a linear to a full
escape graph, we focus on its derivation.

At time $T_i$, all variants $v' \in \alle_{i+1}$ have
$e_{v'}(T_i) \approx P_{v'} \mu \E^2 \delta^2$.  Consequently, since $\E
e_{v'}(T_i) \to \infty$ the $e_{v'}$ dynamics are deterministic in the
\spls and we have,
\begin{equation} \label{E:v_ratio}
\frac{e_{v'}(t)}{\mu^2 \E \delta^2} \approx P_{v'} I_{v'}(t)
\end{equation}
where
\begin{equation}
I_{v'}(t) = \exp[\int^{t}_{T_{i}} ds (\gamma h(s) - k_{i+1}].
\end{equation}
Notice that $I_{v'}(t)$ is dependent only on the variant
class $\alle_{i+1}$ through $k_{i+1}$ but not on the specific variant $v'$ in
$\alle_{i+1}$.  This leads to the ratio for $v', v'' \in \alle_{i+1}$,
\begin{equation}  \label{E:full_e_ratio}
\frac{e_{v'}(T_{i+1})}{e_{v''}(T_{i+1})} =
\frac{P_{v'}}{P_{v''}}.
\end{equation}

Recall that $T_{i+1}$ is defined as the first time for which some variant $v'
\in \alle_{i+1}$ is of scaled population size $\delta$.   Let $v_{\max, i+1}$
be that variant.  Then, by definition
\begin{equation}
e_{v_{\max, i+1}}(T_{i+1}) = \delta,
\end{equation}
Further, from (\ref{E:v_ratio}) we know that $P_{v_{\max, i+1}} = P_{\max,
i+1}$.  Plugging the equation directly above into (\ref{E:full_e_ratio}) with
$v'' = v_{\max,i+1}$ gives
\begin{equation}  \label{E:full_e}
e_{v'}(T_{i+1}) =
\left(\frac{P_{v'}}{P_{\max, i+1}}\right) \delta.
\end{equation}

The $v'$ population size at $T_i$ is too small
to create $v' \to v$ mutations (this is the content of conclusions $1$ and $2$
in Lemma \ref{L:transition_phase}).   Consequently, if we want to know how many
$v$ variants are produced at $T_{i+1}$ by $v' \to v$ mutations and their
descendants, all we have to know is $e_{v'}(T_{i+1})$.  Indeed, the arguments of
section \ref{S:linear_escape_graph} show that when $e_{v'}(T_{i+1}) = \delta$,
then $e_v(T_{i+1}) = \mu^2 \E \delta^2 \S$.   But for the
full escape graph, $e_{v'}(T_{i+1})$ is given by (\ref{E:full_e}) and
consequently we must replace $\delta$ by $(P_{v'}/P_{\max, i+1})\delta$ in
(\ref{E:P_v}).  This substitution gives  (\ref{E:full_pop_partial}).

\subsection{Lineage Construction}  \label{S:full_lineage_construction}

Lemmas \ref{L:lineage_prior_bottle} and \ref{L:lineage_bottle} proved for the
linear escape graph apply to the full escape graph with almost identical
proofs.  Lemma \ref{L:coal_prob}, on the other hand, requires significant
generalization.

In the case of the linear escape graph, given an $A$ we defined $T_A$ on the
$[T_i, T_{i+1}]$ interval as the time at which the number of $v_{i+1} \to
v_{i+2}$ mutations has mean $A$.  More precisely, $T_A$ was defined by,
\begin{equation}
\int_{T_i}^{T_A} ds (\mu \gamma h \E) e_{v_{i+1}}(s) = A.
\end{equation}
For the full escape graph, we generalize the definition of $T_A$ by considering
the variant $v_{\max,i+1}$ (defined in the previous subsection): 
\begin{equation}
\int_{T_i}^{T_A} ds (\mu \gamma h \E) e_{v_{\max,i+1}}(s) = A.
\end{equation}
Using  (\ref{E:full_e_ratio}) with $T_{i+1}$ replaced by $T_A$ we find that the
mean number of $v' \to v$ mutations produced by $T_A$ is
$A(P_{v'}/P(v_{\max,i+1}))$.  In other words, each $v'$ has an associated scaled
version of $A$.  

For the linear escape graph, we sampled $K \sim \text{Poisson}(A)$  versions of
$\Gamma^{(i)}$ and then applied the paintbox construction given by Definition
\ref{D:linear_xi}.   For the full escape graph, the idea is similar except that
for each $v'$ such that $v' \to v$ mutations are possible we must take
$\text{Poisson}(A(P_{v'}/P(v_{\max,i+1}))$ samples of $\Gamma^{(i)}$.  Then, we
must assign different 'colors' for each such sample for each possible $v'$.  

Everything then proceeds according to a paintbox partition, except that the
colors also tell us which variant in the $\alle_{i+1}$ class produced the
lineage being colored.   So, as described by Definition \ref{D:full_xi}, we must
not only coalesce lineages according to the paintbox construction, we must also
allocate the lineages to the appropriate $\alle_{i+1}$ variants at time $T_i$.

To implement all this, we estimate pop values through the decomposition $D_v$
given in (\ref{E:pop_decomp_full}).  As we mentioned directly below the proof of
Lemma \ref{L:coal_prob}, we do not sample pop values directly using
(\ref{E:full_key}) because then we would need to consider conditional
distributions from which we do not know how to sample.   Notice that $K_{v' \to
v}$ depends only on the ratio of pop values,  and so we may use the $D_v$ to
form $K_{v' \to v}$.  Since 
we take all $\alle_1$ variants to be of equal population size at $t=0$, we
may take $D_v = 1$ for $v \in \alle_1$ to start the iteration.   Once $D_v$
and the $\Gamma^{(i)}$ have been sampled according to Definition \ref{D:pop}, we
move backwards from $T_\sample$ to $T_0$  and implement the paintbox
construction of the $\Xi_{A,i}$, Definition \ref{D:full_xi}, in order to
coalesce lineages.

\appendix
\section{Appendix}
\renewcommand{\theequation}{\thesubsection.\arabic{equation}}
\renewcommand{\theclaim}{A.\arabic{claim}}
\renewcommand{\thelemma}{A.\arabic{lemma}}
\renewcommand{\theproposition}{A.\arabic{proposition}}
\setcounter{equation}{0}
\setcounter{lemma}{0}
\setcounter{proposition}{0}
\setcounter{claim}{0}

\subsection{Proof of Lemma \ref{L:transition_phase}}  \label{AS:transition}

In this section, we provide the technical details that support conclusions $3$
and $4$ of Lemma \ref{L:transition_phase}.  From conclusions $1$ and $2$, we
know $e_i, e_{i+1} \ll \delta$ and we can reduce \systems to the following,
\begin{gather} \label{E:three_system}
\frac{dh}{dt} = g \left(1 - h -  h(e_i + e_{i-1} + o(\delta)) \right) \\ \notag 
\frac{de_i}{dt} = \gamma e_i (h - \frac{k_i}{\gamma}) + O(\mu) \\ \notag
\frac{de_{i-1}}{dt} = \gamma e_{i-1} (h - \frac{k_{i-1}}{\gamma}) + O(\mu).
\end{gather}
The $O(\mu)$ terms in the last two equations directly above can be ignored
because $T_i^h - T_i = O(|\log(\delta)|)$ and $\mu |\log(\delta)| \to 0$. 
Dropping these terms, we note the following relation,  
\begin{equation}
\frac{d\big(\log(e_i) - \log(e_{i-1})\big)}{dt} = \delk.
\end{equation}
Integrating the above equation and using our assumptions on $e_{i-1}(T_i)$ and
$e_i(T_i)$ we find,
\begin{equation}  \label{E:mini_lyaponov}
\frac{e_i(t)}{e_{i-1}(t)} = O(\delta \exp[\delk(t - T_i)]).
\end{equation}

If we can show that 
$e_i$ is bounded then as $t$ grows, (\ref{E:mini_lyaponov})
implies that $e_{i-1}$ collapses.   To see that $e_i$ is bounded, set
\begin{equation}
z(t) = \frac{h(t)}{g} + \frac{e_i(t)}{\gamma} + \frac{e_{i-1}(t)}{\gamma}
\end{equation}
Then by straightforward differentiation,
\begin{align}
\frac{dz}{dt}
	& = 1 - h - h \cdot o(\delta) - e_i \frac{k_i}{\gamma} -
	     e_{i-1}\frac{k_{i-1}}{\gamma}
\\ \notag
	& \le 1 - \min(g(1-o(\delta)), k_i, k_{i-1})z
\end{align}
Since $z(t)$ is non-negative, we find that $z(t)$ must be bounded.   In turn
$h$, $e_i$ and $e_{i-1}$ must be bounded.  Returning to (\ref{E:mini_lyaponov})
and setting $t \ge 2/\delk|\log(\delta)|$ we find, since $e_i$ is bounded, 
\begin{equation}
e_{i-1}(t) = O(\delta).
\end{equation}

Once $t > 2/\delk|\log(\delta)|$, we can further reduce (\ref{E:three_system})
to,
\begin{gather} \label{E:two_system}
\frac{dh}{dt} = g \left(1 - h -  h(e_i + O(\delta)) \right) \\ \notag 
\frac{de_i}{dt} = \gamma e_i (h - \frac{k_i}{\gamma}) 
\end{gather}
Consider then (\ref{E:two_system}). Ignoring the $O(\delta)$ term for a moment,
the system is not dependent on $\delta$.  Since we have shown $h, e_i$ to be bounded,
application of Poincare-Bendixon shows that the system converges to its
non-trivial equilibirum, $h = \frac{k_i}{\gamma}$ and $e_i = \frac{1-h}{h}$.  
Now consider the $O(\delta)$ term. Given some fixed distance $\epsilon >
0$, if we run the system from $t = 2/\delk|\log(\delta)|$ to $t =
3/\delk|\log(\delta)|$ we are guaranteed by choosing $\delta$ sufficiently small
to be within $\epsilon$ of the equilibrium.   In turn, taking $\epsilon$ small, we can
linearize (\ref{E:two_system}) about its equilibrium.

Straightforward computation shows that both eigenvalues of the linearized system
have negative real part bounded above by,
\begin{equation} \label{E:eigconstant}
\eigconstant = -(g\gamma + O(\delta)) \min(1, \frac{4k^2}{g\gamma}(1 -
\frac{k}{\gamma}))
\end{equation}
Running the system from $t = 3/\delk|\log(\delta)|$ to $t =
(3 + \frac{2}{|\eigconstant|})|\log(\delta)|$ forces (\ref{E:two_system}) to
within $O(\delta)$ of the equilibrium.

\subsection{Proof of (\ref{E:linear_pop_value_asymp})}  \label{AS:spawning}

In this section we provide technical details that justify
(\ref{E:linear_pop_value_asymp}).  (\ref{E:linear_pop_value_asymp}) is
what underlies our pop value formulas.    The
arguments we employ are extensions of those found in
\cite{Mohle_2005_J_App_Prob}  which considered the $\LD_\classic(t)$  process. 
To make this connection more  explicit, where possible we adopt the notation  of
\cite{Mohle_2005_J_App_Prob}.

Following from the arguments made directly above (\ref{E:immigration_rate}), we
assume the following:
\begin{enumerate}
\item the $v_{i+1}$ population expands deterministically  from time $T_i^h$ to
$T_{i+1}$ with rate $\delk + O(\delta)$
\item at time $T_i^h$ no  $v_{i+2}$ variants exist
\item $v_{i+1} \to v_{i+2}$ mutations occur at rate  $\mu k_i E e_{i+1}$. 
\item $v_{i+2}$ variants have birth and death rates $b=k_i + O(\delta)$ and $d
= k_{i+2} + O(\delta)$.
\end{enumerate}
We set $r = b-d$ and by our assumptions on CTL attack, $r = 2\delk + O(\delta)$.
We will develop our formulas for arbitrary $b,d$ assuming only $r > \delk$.

As observed in \cite{Mohle_2005_J_App_Prob}, the number of mutations is Poisson
distributed with mean $m$ given by,
\begin{equation}  \label{E:appendix_m_general}
m = \int_{T_i^h}^{T_{i+1}} ds \mu k_i \E e_{i+1}(T_i^h) \exp[(\delk +
O(\delta)(s - T_{i+1}))]
\end{equation}
By definition $e_{i+1}(T_{i+1}) = \delta$.  We can then estimate the integral
expression for $m$ and find,
\begin{equation}  \label{E:appendix_m}
m = (\frac{k_i}{\delk})\mu \E \delta \exp[O(\delta)(T_{i+1} - T_i^h)].
\end{equation}
We want to show that $\exp[O(\delta)(T_{i+1} - T_i^h)] \to 1$.   Note the
following three facts,
\begin{itemize}
  \item $e_{i+1}(T_{i+1}) = \delta$
  \item $e_{i+1}(T_i^h) > e_{i+1}(T_i) = O(\mu^2 \E \delta^2)$
  \item On $[T_i^h, T_i]$, $e_{i+1}$ grows deterministically with rate $\delk +
  O(\delta)$,
\end{itemize}
In other words, we know where $e_{i+1}$ starts and ends, and we know its growth
rate.  Then, we can find the bound
\begin{equation}
T_{i+1} - T_i \le O(\frac{1}{|\log(\mu^2 E \delta^2)|}).
\end{equation}
And so by our definition of $\delta$, (\ref{E:delta_def}), we can conclude
$\exp[O(\delta)(T_{i+1} - T_i^h)] \to 1$.  Plugging this observation into
(\ref{E:appendix_m}) gives,
\begin{equation} \label{E:appendix_m_final}
\frac{m}{(\frac{k_i}{\delk})\mu \E \delta} \to 1.
\end{equation}
A similar argument will show that the $O(\delta)$ expressions in the formulas
for $b,d$ and mutation rate will have no effect in the \spl.    For simplicity
then, we drop $O(\delta)$ terms from this point on without comment.     
 
Connecting to the notation of \cite{Mohle_2005_J_App_Prob}, we let $Y_m =
e_{i+2}^\#(T_{i+1})$, the number of $v_{i+2}$ variants at time $T_{i+1}$. We can
decompose $Y_m$ by
\begin{equation}  \label{E:Y_m_def}
Y_m = \sum_{k=1}^K X_k,
\end{equation}
where $K$ is Poisson distributed with mean $m$, given in (\ref{E:appendix_m}),
and each $X_k$ represents  the number of $v_{i+2}$ variants at time $T_{i+1}$
that descend from a single $v_{i+2}$ variant produced by a $v_{i+1} \to v_{i+2}$
mutation (compare (2.1) in \cite{Mohle_2005_J_App_Prob}).   Conditioned on $K$,
the times of the $K$ mutations are iid. and in turn, the $X_k$ are iid.  

The following lemma characterizes the Laplace transform of $X_k$.  

\begin{lemma} \label{L:A_X}
Let the $X_k$ be iid versions of the r.v. $X$.  Recall $r=b-d$.  Then
\begin{equation}
\lim_{\lambda \to 0} \frac{E[\exp[-\lambda X]] - 1}{\lambda^\frac{\delk}{r}} =
-\Upsilon(b,d,\delk)
\end{equation}
where
\begin{equation}  \label{E:upsilon_def}
\Upsilon(b,d,\delk) =  -\frac{r}{r-\delk}(\frac{\delk}{r}-1)(\frac{\delk}{r})
		\pi \left(\frac{r}{b}\right)^{1 - \frac{\delk}{r}}
		\csc(\pi \frac{\delk}{r})
\end{equation}
\end{lemma}

Before proving Lemma \ref{L:A_X}, we state and demonstrate Proposition
\ref{P:A_LD}  which characterizes the distribution of $Y_m$ for large $m$. 
(\ref{E:Mohle_LD_2}) follows from Proposition \ref{P:A_LD} by setting $r =
2\delk$.   (\ref{E:linear_pop_value_asymp}) follows from the proposition by
further  setting $m = (k_i/\delk)\mu \E \delta$ as justified by
(\ref{E:appendix_m_final}).   

\begin{proposition}  \label{P:A_LD}
\begin{equation}
\lim_{m \to \infty} \frac{Y_m}{m^\frac{r}{\delk}} \to \S(\frac{\delk}{r}, 1,
\Upsilon(r,\delk))
\end{equation}
\end{proposition}

To see that Proposition \ref{P:A_LD} follows from Lemma \ref{L:A_X}, first
notice that by (\ref{E:Y_m_def}) and the independence
of the $X_k$ we have
\begin{equation}
E[-\lambda Y_m] = \exp[m(E[\exp[-\lambda X]] - 1)].
\end{equation}
Replacing $\lambda$ by $\lambda/m^{r/\delk}$ in the above equality and applying
Lemma \ref{L:A_X} gives,
\begin{align}  \label{E:Y_m_laplace_limit}
\lim_{m \to \infty} E[-\lambda \frac{Y_m}{m^\frac{r}{\delk}}] =
\exp[-\lambda^\frac{\delk}{r} \Upsilon(b,d,\delk)].
\end{align}
Stable distributions are characterized by index $\alpha$,
 the skewness parameter $\beta$, and the scale factor $c$  \cite{Nolan_book}. 
The limit above is seen as the Laplace transform of a $\alpha = \frac{\delk}{r}$
stable distribution.  Since the support of $Y_m$ is on $[0, \infty]$, $\beta
= 1$ \cite{Nolan_book}.  $c$ is defined through the characteristic function of
the stable process.   To  determine $c$ we invert the characteristic function of
$S(\delk/r, 1, c)$ and  compute its Laplace transform.   Setting this result
equal to the right side of (\ref{E:Y_m_laplace_limit}) we find,
\begin{equation}
c = \left(\Upsilon(b,d,\delk)
    \frac{1 + \cos(\pi \alpha)}{\cos(\frac{\pi\alpha}{2})}
    \right)^\frac{1}{\alpha}
\end{equation}
In the case $\alpha = 1/2$ we find
\begin{equation}
c = 2 \Upsilon^2(b,d,\delk) 
\end{equation}
If we plug in our values for $b,d$ in the formula for $\Upsilon$ we find
\begin{equation}
\Upsilon = -\frac{\pi}{2}\sqrt{\frac{2\delk}{k_i}},
\end{equation}
and so we arrive at
\begin{equation}
c = \pi^2\left(\frac{\delk}{k_i}\right).
\end{equation}

Finally, we give the proof of Lemma \ref{L:A_X}.
\begin{proof}[Proof of Lemma \ref{L:A_X}] 
Let $t^*$ be the random
time of a single mutation.   A standard Poisson process argument applied to
(\ref{E:appendix_m_general}) gives for the density of $t^*$,
\begin{align}
P(t^* = s) & =   \delk \exp[-\delk (T_{i+1} - s)],
\end{align}
Let $Z(t)$ be
a branching process with birth and death rates $b,d$
respectively run to time $t$ assuming $Z(0) = 1$.  Then $X = Z(t^*)$ and we can
represent $E[\exp[-\lambda X]] - 1$ by,
\begin{equation}  \label{E:Mohle_I}
E[\exp[-\lambda X]] - 1 = 
	\int_{T_i^h}^{T_{i+1}} ds \delk \exp[-\delk (T_{i+1} - s)] G(T_{i+1}-s),
\end{equation}
where 
\begin{equation}
G(t) = \exp[-\lambda Z(t)] - 1.
\end{equation}

Extending the range of integration from $T_i^h$ down to $-\infty$ also adds
error that collapses in the \spl.   With this observation and the change of
variable $s \to T_{i+1} - s$ we consider,
\begin{equation}  \label{E:Mohle_II}
E[\exp[-\lambda X]] - 1 = \int_0^\infty ds \delk \exp[-\delk s] G(s)
\end{equation}

To understand  (\ref{E:Mohle_II}), we have found it useful to make
the substitution $w = \lambda \exp[r s]$.  To explain why, we
observe the well known fact, $E[Z(t)] = \exp[rt]$.  We expect then,
\begin{equation} 
G(t) \approx \exp[-\lambda \exp[rt]] - 1 = \exp[-w] - 1.
\end{equation}  
As a result, at least for us, analyzing (\ref{E:Mohle_II}) in terms of $w$
simplifies the contribution of the $G(t)$ term to (\ref{E:Mohle_II}).   Making
the substitution gives,
\begin{align}  \label{E:Mohle_III}
E[\exp[-\lambda X]] - 1 =
 \rr  \lambda^\rr \int_\lambda^\infty dw \frac{1}{w^{1 + \rr}} G(f(w))
\end{align}
where $f(w) = \frac{1}{r}\log(\frac{w}{\lambda})$.  The integral to the
right of the equality directly above has a limit as $\lambda \to 0$.   This
is not obvious due to the singularity of $\frac{1}{w^{1 + \rr}}$ at $w=0$.  To
remove the singularity and show that the integral is indeed $O(1)$ we apply
partial integration twice, integrating the $\frac{1}{w^{1 + \rr}}$ term and
differentiating the $G(f)$ term.   We find,
\begin{align}  \label{E:final_Mohle}
E[\exp[-\lambda X]] - 1 &= \lambda^\rr \bigg[
	-\frac{r}{r - \delk} \int_\lambda^\infty dw (w^{1 - \rr}) [G(f(w))]''
	+ I_1 + I_2 \bigg]
\end{align}
where
\begin{gather}
I_1 = - w^{-\rr} G(f(w))\bigg|_\lambda^\infty, \\ \notag
I_2 =  (\frac{r}{r - \delk}) w^{1 - \rr} [G(f(w))]'\bigg|_\lambda^\infty.
\end{gather}

$G$ can be characterized through a Kolmogorov equation.   In our specific case
\begin{equation}  \label{E:Kol}
G' = b G(G + (1 - \frac{d}{b})).
\end{equation}
(\ref{E:Kol}) can be integrated to find a specific formula
for $G$ and using (\ref{E:final_Mohle}) we find
\begin{align}  \label{E:final_Mohle_2}
\lim_{\lambda \to 0} \frac{E[\exp[-\lambda X]] - 1}{\lambda^\frac{\delk}{r}} &=
\Upsilon(b,d,\delk),
\end{align}
\end{proof}

\subsection{Proof of Lineage Construction Approximation}  \label{AS:coal_prob}

In this subsection, we demonstrate the two assertions used to prove Lemma
\ref{L:coal_prob} which justifies our paintbox construction approximation of
$\Pi(0)$.

\begin{lemma}
\begin{equation} \label{AE:ratio}  
\frac{\sum_{q' \in \M_\error} \E e_{i+1}^{(q')}(T_{i+1})}
{\sum_{q \in \M_A} \E e_{i+1}^{(q)}(T_{i+1})}
 = O(\frac{1}{A}),
\end{equation}
\end{lemma}

\begin{proof}
If we sum the numerator and denominator of (\ref{AE:ratio}), we have $Y_m$ which
was introduced in the previous section, see (\ref{E:Y_m_def}).  We know
through Proposition \ref{P:A_LD} that $Y_m/m^2$ converges to a stable
distribution. We now apply the arguments of Proposition \ref{P:A_LD}
separately to the numerator and denominator.   The difference in the arguments will be that the bounds
of integration in  (\ref{E:Mohle_III}) will no longer be $\lambda$ and $\infty$. 
However the same essential ideas are used and we find,
\begin{align}  \label{E:num}
\sum_{q \in \M_A} \E e_{i+2}^{(q)}(T_{i+1}) 
	\approx \left(\frac{\delk}{k_i}\right)^2 (\delta \mu \E)^2 \S(\frac{1}{2}, 1,
	\pi^2\left(\frac{\delk}{k_i}\right))
\end{align}
while
\begin{align}  \label{E:den}
\sum_{q' \in \M_\error} \E e_{i+2}^{(q')}(T_{i+1}) 
	\approx O((\delta \mu \E)^2) \Gauss(\frac{1}{A}, \frac{1}{A^3}),
\end{align}
where $\Gauss(\mu, \sigma^2)$ is a normal distribution with mean $\mu$ and
variance $\sigma^2$.  Notice that the early mutations produce a heavy tailed
distribution, while the later mutations are normally distributed.  By time
$T_A$, the rate of mutation is $O(A)$, as a result many mutations happen at
approximately the same time and through an appropriate scaling these nearly
simultaneous mutations lead to a central limit theorem.  Taking the ratio of
(\ref{E:num}) and (\ref{E:den}) gives (\ref{AE:ratio})  
\end{proof}

\begin{lemma}
\begin{equation}  \label{AE:q_limit}
\frac{\E e_{i+2}^{(q)}}{ \sum_{\M_A} \E e_{i+2}^{(q')}} \to 
\frac{\Gamma_q}{\sum_{q' \in \M_A} \Gamma_{q'}},
\end{equation}
where the $\Gamma_q$ are sampled from $\Gamma^{(i)}$.  
\end{lemma}
 
\begin{proof}
Let $t^*$ be the random time of a mutation on $[T_i^h, T_A]$.   Up to an error
that disappears in the limit, we have
\begin{equation}
P(T_A - t^* = s) = \delk \exp[-\delk s].
\end{equation}
In other words $T_A - t^*$ is exponentially distributed with rate $\delk$ and
setting
\begin{equation}
\xi_{q,2} = \delk(T_A - t^*)
\end{equation}
defines $\xi_{q,2}$ as a exponential r.v. with rate $1$.   When a
$v_{i+1} \to v_{i+2}$ mutation occurs at time $t^*$, the 
number of  descendants at $T_{i+1}$, $E e_{i+2}^{(q)}(T_{i+1})$ 
is given  by $Z(T_{i+1} - t^*)$ (recall the definition
of $Z(t)$ from section \ref{AS:spawning}).   Since $T_{i+1} - T_A \to \infty$
and $t^* < T_A$, we have $T_{i+1} - t^* \to \infty$.  Standard asymptotic
results from branching process theory (see \cite{Athreya_Book}) give,
\begin{equation}
\exp[-2\delk(T_i - t^*)] Z(T_i - t^*) \to \frac{k_{i}}{2\delk} \xi_{q,1}
B(\frac{2\delk}{k_{i}}),
\end{equation}
where $\xi_{q,1}$ is exponential with rate $1$ and $B(p)$ is a Bernoulli r.v.
with success probability $p$. We can apply this analysis to each $\E
e_{i+2}^{(q)}$ in (\ref{AE:q_limit}) by multiplying the numerator and
denominator by $\exp[-2\delk(T_i-T_A)]$ to find,
\begin{align}
\frac{\E e_{i+2}^{(q)}(T_{i+1})}{ \sum_{q \in \M_A} \E e_{i+2}^{(q')}(T_{i+1})}
& \to \frac{\exp[2\delk(T_A - t^*_q)] 
		\frac{k_{i}}{2\delk} \xi_{q,1} B_q(\frac{2\delk}{k_{i}})}
      {\sum_{q' \in \M_A} \exp[2\delk(T_A - t^*_{q'})] \frac{k_{i}}{2\delk}
       \xi_{q',1} B_{q'}(\frac{2\delk}{k_{i}})} 
\\ \notag
	& = \frac{\exp[2 \zeta_{q,2}] \xi_{q,1} B_q(\frac{2\delk}{k_{i}})}
	{\sum_{q' \in \M_A} \exp[2 \xi_{q',2}] \xi_{q',1}
	B_{q'}(\frac{2\delk}{k_{i}})}
\\ \notag
	& = \frac{\Gamma_q}{\sum_{q' \in \M_A}  \Gamma_{q'}}.
\end{align}
\end{proof}


\newcommand{\noopsort}[1]{} \newcommand{\printfirst}[2]{#1}
  \newcommand{\singleletter}[1]{#1} \newcommand{\switchargs}[2]{#2#1}

\end{document}